\documentclass[12pt]{article}

\usepackage{a4wide, amsmath,amsthm,amsfonts,amscd,amssymb,eucal,bbm,mathrsfs}

\def\RR{{\mathbb R}}
\def\CC{{\mathbb C}}
\def\NN{{\mathbb N}}
\def\ZZ{{\mathbb Z}}

\def\A{{\mathcal A}}
\def\B{{\mathcal B}}

\def\E{{\mathcal E}}

\def\H{{\mathcal H}}
\def\I{{\mathcal I}}
\def\K{{\mathcal K}}
\def\M{{\mathcal M}}
\def\N{{\mathcal N}}
\def\O{{\mathcal O}}

\def\R{{\mathcal R}}
\def\S{{\mathcal S}}

\def\U{{\mathcal U}}

\def\a{\alpha}
\def\b{\beta}
\def\d{\delta}
\def\e{\varepsilon}

\def\l{\lambda}
\def\L{\Lambda}

\def\o{\omega}

\def\r{\rho}
\def\s{\sigma}

\def\t{\tau}

\def\w{\omega}

\def\gA{\mathfrak A}

\def\gM{\mathfrak M}

\def\Ad{{\hbox{\rm Ad}}}
\def\Aut{{\hbox{Aut}}}
\def\Out{{\hbox{Out}}}

\def\id{{\rm id}}
\def\1{{\mathbbm 1}}
\def\Exp{{\rm Exp}}

\def\geo{\varphi_{\rm geo}}
\def\u1{U(1)}
\def\intervals{\mathcal{I}}
\def\diff{{\rm Diff}}
\def\Diff{{\rm Diff}}
\def\diffs1{\diff(S^1)}
\def\vir{{\rm Vir}}

\def\supp{{\rm supp}}
\def\psl2r{{\rm PSL}(2,\RR)}
\def\<{\langle}
\def\>{\rangle}

\newtheorem{theorem}{Theorem}[section]
\newtheorem{definition}[theorem]{Definition}
\newtheorem{corollary}[theorem]{Corollary}
\newtheorem{proposition}[theorem]{Proposition}
\newtheorem{lemma}[theorem]{Lemma}
\theoremstyle{remark}
\newtheorem{rem}[theorem]{Remark}

\begin{document}
\date{}

\title{\huge{Thermal States in Conformal QFT. I}}

\author{\textsc{\normalsize Paolo Camassa, Roberto Longo, Yoh Tanimoto, Mih\'aly
Weiner\footnote{Permanent address: Alfr\'ed R\'enyi Institute of Mathematics H-1364 Budapest, POB 127, Hungary}}\\  {\normalsize Universit\`a di Roma ``Tor Vergata'',
Dipartimento di Matematica}
\\  {\normalsize Via della Ricerca Scientifica, 1 - 00133 Roma, Italy}}

\maketitle
\begin{abstract}
We analyze the set of locally normal KMS states w.r.t. the translation group for a local conformal net $\A$ of von Neumann algebras on $\mathbb R$. In this first part, we focus on completely rational net $\A$. Our main result here states that, if $\A$ is completely rational, there exists exactly one locally normal KMS state $\varphi$. Moreover, $\varphi$ is canonically constructed by a geometric procedure. A crucial r\^ole is played by the analysis of  the ``thermal completion net'' associated with a locally normal KMS state.
A similar uniqueness result holds for KMS states of two-dimensional local conformal nets w.r.t. the time-translation one-parameter group.
\end{abstract}
\vskip 6.8cm

\noindent{\footnotesize Research supported in part by the ERC Advanced Grant 227458
OACFT ``Operator Algebras and Conformal Field Theory", PRIN-MIUR, GNAMPA-INDAM and EU network ``Noncommutative Geometry" MRTN-CT-2006-0031962.}
\vskip 0.5cm
\noindent{\footnotesize Email: camassa@mat.uniroma2.it, longo@mat.uniroma2.it,
tanimoto@mat.uniroma2.it, mweiner@renyi.hu}
\newpage

\section{Introduction}\label{introduction}

Although Quantum Field Theory is primarily designed to study finitely many particle states, the thermal aspects in QFT are of crucial importance for various reasons and one naturally aims at
a general analysis of the thermal behavior starting from the basic properties shared by any QFT. As is known, at infinite volume the thermal equilibrium states are characterized by the Kubo-Martin-Schwinger condition (see \cite{H}), in other words KMS states are Gibbs states for infinite volume systems. A model independent construction of KMS states in QFT has been achieved in \cite{BJ} for QFT nets of $C^*$-algebras with the natural nuclearity property; the constructed states, however, are not necessarily locally normal, i.e. the restrictions of these KMS states to bounded spacetime regions are not associated with the vacuum representation.

We now mention that, among other motivations to study thermal states in QFT, an important one comes from cosmological considerations and in quantum black hole analysis, in particular concerning the Hawking-Unruh thermal radiation. An elementary situation where this can be illustrated is the Schwarzschild black hole case where the restriction of the vacuum state on the full Schwarzschild-Kruskal spacetime to the Schwarzschild spacetime algebra satisfies the KMS condition at Hawking temperature. This example also suggests the Operator Algebraic approach to be the natural one in this thermal analysis, indeed the Bisognano-Wichmann theorem provides a model independent derivation of this result. After all, the definition of a KMS state concerns a $C^*$-algebraic dynamical system.

In this work we initiate a general study of thermal states in CFT (conformal QFT), more precisely of the locally normal KMS states with respect to the translation one-parameter group. There are several motivations for us to focus our attention to low dimensional CFT, for example in the mentioned black hole context one gets a conformal net by restriction to the horizon (see \cite{GLRV}), but also because CFT represents a limit case of general QFT; moreover conformal nets naturally arise and play a crucial role in different mathematical and physical contexts.

Before explaining our result, we wish to recall the general Quantum Statistical Mechanics description of the chemical potential made in \cite{AHKT}, where the the chemical potential of a KMS state on the observable algebra turns to appear as a label for the different extremal KMS states on the field algebra. Here the observable algebra is the fixed-point algebra w.r.t. a compact gauge group. A similar structure appears in QFT on the four-dimensional Minkowski spacetime, where the main difference occurs because in QFT one deals with a net of local von Neumann algebras and different topologies are involved. One may extends, although not obviously, the results in \cite{AHKT} to the QFT framework and indeed we shall need and provide this extension at some point for model analysis in the second part of this paper.

Yet, for the general study of KMS states in chiral CFT the results in \cite{AHKT} are totally not applicable because there is no field algebra due to the occurrence of braid group statistics. Therefore, a completely different approach was proposed in \cite{SW} and
studied in detail in \cite{Longo}, making a crucial use of the conformal symmetries.
Starting with a local conformal net $\A$ of von Neumann algebras on the real line and
a KMS state $\varphi$ on $\A$ w.r.t. translations, a new local, M\"obius covariant net was
constructed, the thermal completion of $\A$ w.r.t. to $\varphi$, playing a main role
in the analysis.

Our main object in this paper is a local conformal net $\A$ of von Neumann algebras on $S^1$:
\[
I\in\I\mapsto \A(I)
\]
where conformal stands for diffeomorphism covariant, $\I$ is the set of intervals of $S^1$ and the $\A(I)$'s are von Neumann algebras on a fixed Hilbert space (see below). Indeed we take the ``real line picture'', namely $\A$ is restricted to the real line $\mathbb R$, where $\mathbb R$ is identified with $S^1\setminus \{-1\}$ by the stereographic map. Then we consider the quasi-local $C^*$-algebra
\[
\gA_\A \equiv \overline{\bigcup_{I\Subset \mathbb R}\A(I)}^{\|\cdot\|}\ .
\]
Here the union is over the bounded open intervals of $\mathbb R$ and the closure is in the norm topology. The translations $\t_s: t\mapsto t+s$
are unitarily implemented by $U(\t_s)$ and give rise to a one-parameter automorphism group
$\Ad U(\tau_s)$ of $\gA_\A$.
Our aim is to study the locally normal KMS states $\varphi$ of $\gA_\A$
w.r.t. $\Ad U(\t)$. We recall that $\varphi$ is KMS at inverse temperature $\beta>0$ if for all $x,y\in\gA_\A$ there is a bounded continuous function $f_{xy}$ on the strip $0\leq\Im z \leq \beta$, analytic in the interior $0<\Im z < \beta$ such that
\[
f_{xy}(t) = \varphi(\Ad U(\tau_t)(x)y),\qquad
f_{xy}(t+ i\beta) = \varphi(y\Ad U(\tau_t)(x))\ .
\]
Now, by the assumed scale invariance, we have a one-parameter
automorphism group $\Ad U(\delta_s)$
of $\gA_\A$ corresponding to the dilations $\d_s:t\mapsto e^s t$, so the state $\varphi$ is KMS at inverse temperature $\beta$ iff the state $\varphi \circ \Ad U(\delta_s)$ is KMS at inverse temperature $e^s \beta$. It follows that the structure of KMS states does not depend on the temperature; in physical terms, there are no phase transitions. For this reason we fix the inverse temperature $\beta = 1$ in the rest of this paper.

Our first observation is that there always exists a canonical KMS state, that is constructed by a geometric procedure. Indeed the restriction of the vacuum state to the von Neumann algebra associated with the positive real line is KMS w.r.t. the (rescaled) dilation group (Bisognano-Wichmann property \cite{BGL,FG}); now the exponential map intertwines translations with dilations and one can use it to pull back the vacuum state and define the geometric KMS state w.r.t. translations.

One may ask whether this geometric KMS state is the only one or there are other
locally normal KMS states (different phases, in physical terms).
Indeed in general there are many KMS states
as we shall see in particular by analyzing the KMS states  of the $U(1)$-current net in the second part of this paper.

We now state our main result: if $\A$ is a completely rational local conformal net,
there exists exactly one locally normal KMS state $\varphi$ with respect to the translation group $\Ad U(\tau)$. Moreover, $\varphi$ is canonically constructed by a geometric procedure.
As we shall see, the proof of this result is obtained in several steps by a crucial use of the thermal completion net and an inductive extension procedure. This is
in accordance with the previous result which showed the uniqueness of
ground state (which is considered as a state with zero temperature)
on loop algebras \cite{Tanimoto}.

Our results extends to the case of a local conformal net $\A$ of von Neumann algebras on the two-dimensional Minkowski spacetime. We shall show that, if $\A$ is completely rational, there exists a unique KMS state w.r.t. the time-translation one-parameter group. Also in this case the KMS state has a geometric origin.

In the second part of this paper we shall study the set of KMS states for local conformal nets that are not rational.

\section{Preliminaries}\label{preliminaries}
\subsection{Conformal QFT on $S^1$}\label{sec:prel1}
Here we exhibit the mathematical setting which we use to describe physical systems
on one-dimensional spacetime $S^1$.
Let $\intervals$ be the set of all open, connected, non-dense, non-empty subsets of $S^1$. We call elements of $\intervals$ {\bf intervals} in $S^1$. For an interval $I$, we denote by $I^\prime$ the interior of the complement $S^1 \setminus I$.
The group $\psl2r$ acts on $S^1$ by the linear fractional transformations.

A {\bf (local) M\"obius covariant net} is an assignment $\A$ to each interval
of a von Neumann algebra $\A(I)$ on a fixed separable Hilbert space $\H$
with the following conditions:
\begin{enumerate}
\item[(1)] {\bf Isotony.} If $I_1 \subset I_2$, then $\A(I_1) \subset \A(I_2)$.\label{isotony}
\item[(2)] {\bf Locality.} If $I_1 \cap I_2 = \emptyset$, then $[\A(I_1),\A(I_2)] = 0$.
\item[(3)] {\bf M\"obius covariance.} There exists a strongly continuous unitary
representation $U$ of the M\"obius group $\psl2r$ such that
for any interval $I$ it holds that
\begin{equation*}
U(g)\A(I)U(g)^* = \A(gI), \mbox{ for } g \in \psl2r.
\end{equation*}
\item[(4)]{\bf Positivity of energy.} The generator of the one-parameter subgroup of
rotations in the representation $U$ is positive.
\item[(5)] {\bf Existence of vacuum.} There is a unique (up to a phase) unit vector $\Omega$ in
$\H$ which is invariant under the action of $U$,
and cyclic for $\bigvee_{I \in \intervals} \A(I)$.
\end{enumerate}

It is well-known that, from these conditions, the following properties automatically follow
(see, for example, \cite{FG}):
\begin{enumerate}
\item[(6)] {\bf Reeh-Schlieder property.} The vector $\Omega$ is cyclic and separating for each $\A(I)$.
\item[(7)] {\bf Haag duality.} For any interval $I$ it holds that $\A(I)^\prime = \A(I^\prime)$.
\item[(8)] {\bf Bisognano-Wichmann property.} The Tomita-Takesaki operator $\Delta_I$ of
$\A(I)$ with respect to $\Omega$ satisfies the following:
\begin{equation*}
U(\d^I(2\pi t)) = \Delta_I^{-it},
\end{equation*}
where $\d^I$ is the one-parameter group in $\psl2r$ which preserves
the interval $I$ (which we call ``the dilation associated to $I$'': in the real line picture $\d^I: x\mapsto e^s x$ if $I\equiv \mathbb R^+$).
\item[(9)] {\bf Factoriality.} Each local algebra $\A(I)$ is a type ${\rm {\!I\!I\!I}}_1$-factor (unless $\H$ is one dimensional).
\end{enumerate}

The Bisognano-Wichmann property is of particular importance in our context.
Precisely, this property means that the vacuum state
$\omega(\cdot) = \<\Omega,\cdot \Omega\>$
is a KMS state for $\A(I)$ with respect to $\d^I$ (at inverse temperature $2\pi$), see below. This will be exploited to construct a standard KMS state with respect to the spacetime translation in Section \ref{geometric}.

\subsection{Subnets and extensions}\label{subnets}
Let $\B$ be a M\"obius covariant net on $\H$. Another assignment $\A$ of von Neumann algebras
$\{\A(I)\}_{I \in \intervals}$ on $\H$ is called a {\bf subnet} of $\B$ if it satisfies
isotony, M\"obius covariance with respect to the same $U$ for $\B$
and it holds that $\A(I) \subset \B(I)$ for every interval $I \in \intervals$.
If $\A(I)^\prime \cap \B(I) = \CC \1$ for an interval $I$ (hence for any interval, by
the covariance and the transitivity of the action of $\psl2r$ on $\intervals$),
we say that the inclusion of nets $\A \subset \B$ is {\bf irreducible}.

Let us denote by $\H_\A$ the subspace of $\H$ generated by $\{\A(I)\}_{I\in\intervals}$ from $\Omega$, and by $P_\A$ the orthogonal projection onto $\H_\A$.
Then it is easy to see that $P_\A$ commutes with all $\A(I)$ and $U$.
The assignment $\{\A(I)\vert_{\H_\A}\}_{I\in\intervals}$ with the representation $U\vert_{\H_\A}$
of $\psl2r$ and the vacuum $\Omega$ is a M\"obius covariant
net on $\H_\A$. Conversely, if a M\"obius covariant net $\A_0$ is unitarily equivalent to
such a restriction $\A|_{\H_\A}$ of a subnet $\A$ of $\B$, then $\B$ is called an
{\bf extension} of $\A_0$. We write simply $\A_0 \subset \B$ if no confusion arises.

When we have an inclusion of nets $\A \subset \B$, for each interval $I$ there is a
canonical conditional expectation $E_I: \A(I) \to \B(I)$ which preserves the vacuum state $\omega$
thanks to the Reeh-Schlieder property and Takesaki's theorem \cite[Theorem IX.4.2]{Takesaki2}.
We define the {\bf index} of the inclusion $\A \subset \B$ as the index $[\B(I),\A(I)]$
with respect to this conditional expectation \cite{Kosaki}, which
does not depend on $I$ (again by covariance, or even without covariance \cite{LR}).
If the index is finite, the inclusion is irreducible.

\subsection{Diffeomorphism covariance and Virasoro nets}\label{diffeomorphismcovariance}
In the present paper we will consider a class of nets with a much
larger group of symmetry, which still contains many interesting examples.
Let $\diffs1$ be the group of orientation-preserving
diffeomorphisms of the circle $S^1$. This group naturally contains $\psl2r$.

A M\"obius covariant net $\A$ is said to be a {\bf conformal net} if the representation
$U$ extends to a projective unitary representation of $\diffs1$ such that
for any interval $I$ and $x \in \A(I)$ it holds that
\begin{gather*}
U(g)\A(I)U(g)^* = \A(gI), \mbox{ for } g \in \diffs1,\\
U(g)xU(g)^* = x, \mbox{ if } \supp(g) \subset I^\prime,
\end{gather*}
where $\supp(g) \subset I^\prime$ means that $g$ acts identically on $I$. In this case we say that $\A$ is {\bf diffeomorphism covariant}.

 From the second equation above we see that $U(g) \in \A(I)$
if $\supp(g) \subset I$ by Haag duality. If we define
\[
\vir(I) = \{U(g): \supp(g) \subset I\}^{\prime\prime},
\]
one can show that $\vir$ is a subnet of $\A$. Such a net is called
a {\bf Virasoro net}. Let us consider its restriction to the space $\H_\vir$.
The representation $U$ of $\diffs1$ restricts to $\H_\vir$ as well, and this restriction is irreducible
by the Haag duality. In addition, the restriction of $U$ to $\psl2r$ admits
an invariant vector $\Omega$ and the rotation still has positive energy.
Such representations have been completely classified by positive numbers
$c$, the {\bf central charge}, see for example \cite[Appendix A]{C}. It is known that even to the full representation $U$
on $\H$ we can assign the central charge $c$. Since the representation $U$ which
makes $\A$ diffeomorphism covariant is unique \cite{CW}, the value of $c$ is
an invariant of $\A$. We say that the net $\A$ has the central charge $c$.

Throughout the present paper, $\A$ is assumed to be diffeomorphism covariant.

\subsection{Complete rationality}\label{completerationality}
We now define the class of conformal nets to which our main result applies.
Let us consider the following conditions on a net $\A$.
 For intervals   $I_1, I_2$, we shall write $I_1\Subset I_2$ if the closure of $I_1$ is contained in the interior of $I_2$.
\begin{itemize}
\item[(a)] {\bf Split property.} For intervals $I_1\Subset I_2$
there exists a type I factor $F$ such that $\A(I_1) \subset F \subset \A(I_2)$.
\item[(b)] {\bf Strong additivity.} For intervals $I,I_1,I_2$ such that
$I_1 \cup I_2 \subset I$, $I_1 \cap I_2 = \emptyset$, and $I \setminus (I_1\cup I_2)$
consists of one point, it holds that $\A(I) = \A(I_1)\vee \A(I_2)$.
\item[(c)] {\bf Finiteness of $\mu$-index.} For disjoint intervals $I_1,I_2,I_3,I_4$
in a clockwise (or counterclockwise) order with a dense union in $S^1$, the Jones index
of the inclusion $\A(I_1)\vee \A(I_3) \subset (\A(I_2)\vee \A(I_4))^\prime$ is finite
(it does not depend on the choice of intervals \cite{KLM}
and we call it the {\bf $\mu$-index} of $\A$).
\end{itemize}
A conformal net $\A$ is said to be {\bf completely rational} if it satisfies
the three conditions above. If $\A$ is diffeomorphism covariant, the strong additivity condition $(b)$ follows from the other two $(a)$ and $(c)$ \cite{LX}.

An important class of completely rational nets is given by the conformal nets with
$c < 1$, which have been completely classified \cite{KL}. Among other examples of completely rational nets
(with $c \geq 1$) are $SU(N)_k$ loop group
nets \cite{FG,Xu}. It is known that complete rationality passes to
finite index extensions and finite index subnets \cite{Longo03}. The importance
of complete rationality is revealed in representation theory of nets
(see Section \ref{prelimsectors}).

\subsection{Representations and sectors of conformal nets}\label{prelimsectors}

Let $\A$ be a conformal net on $S^1$. A {\bf representation} $\pi$ of $\A$
is a family of (normal) representations $\pi_I$ of algebras $\A(I)$ on a
common Hilbert space $\H_\pi$ with the consistency condition
\[ \pi_J|_{\A(I)} = \pi_I, \mbox{ for } I \subset J. \]
A representation $\pi$ satisfying
$\{\cup_I \pi_I(\A(I))\}'=\CC \1$ is called {\bf irreducible}.
Two representations $\pi,\pi^\prime$
are {\bf unitarily equivalent} iff there is a unitary operator $W$ such
that $\Ad(W)\circ \pi_I = \pi^\prime_I$ for every interval $I$. A unitary
equivalence class of an irreducible representations is called a {\bf
sector}. It is known
that any completely rational net admits only finitely many sectors
\cite{KLM}.

A representation may be given also on the original (vacuum-)Hilbert
space.
Such a representation $\rho$ which preserves each local algebra
$\A(I)$ is
called an {\bf endomorphism} of $\A$. Note that this notion of
endomorphisms differs from the terminology of localized endomorphisms of DHR representation
theory, in which not all local algebras are preserved.
If each representation of the local algebra is surjective, it is called
an {\bf automorphism}. An automorphism which
preserves the vacuum state is called
an {\bf inner symmetry}. Any inner symmetry is implemented by
a unitary operator and it is in the same sector as the vacuum
representation.

\subsection{The restriction of a net to the real line}\label{restriction}
Although conformal nets are defined on the circle $S^1$, it is natural
from a physical point of view to consider a theory on the real line $\RR$.
We identify $\RR$ with the punctured circle $S^1\setminus \{-1\}$ by the
Cayley transform:
\[
t = i\frac{1+z}{1-z} \Longleftrightarrow z = \frac{t-i}{t+i}, t \in \RR, z \in S^1 \subset \CC.
\]
The point $-1 \in S^1$ is referred to as ``the point at infinity'' $\infty$ when
considered in the real-line picture.

We recall that the M\"obius group $\psl2r$ is generated by the following
three one-parameter groups, namely rotations, translations and dilations \cite{Longo08}:
\begin{eqnarray*}
\rho_s(z) &=& e^{is}z, \mbox{ for } z \in S^1 \subset \mathbb{C}\\
\tau_s(t) &=& t + s, \mbox{ for } t \in \mathbb{R} \\
\delta_s(t) &=& e^{s}t, \mbox{ for } t \in \mathbb{R},
\end{eqnarray*}
where rotations are defined in the circle picture, on the other hand translations and
dilations are defined in the real line picture. Of these, translations and dilations
do not move the point at infinity.

According to this identification, we also restrict a conformal net $\A$ to
the real line. Namely, we consider all the finite-length open intervals
$I \Subset \RR = S^1\setminus \{-1\}$
under the identification. We still have an isotonic and local net of von Neumann algebras
corresponding to intervals in $\RR$, which is covariant under translation, dilation and
diffeomorphisms of $S^1$ which preserve $-1$. It is known that the positivity of energy
(the generator of rotations) is equivalent to the positivity of the generator of
translations \cite{Weiner05}, and the vacuum vector $\Omega$ is invariant under translations and
dilations. We denote this restriction to the real line by $\A|_\RR$.

The terminology
of representations easily translates to the real-line picture. Namely, a representation
of $\A|_\RR$ is a consistent family $\{\pi_I\}_{I \Subset \RR}$ of representations of
$\{\A(I)\}_{I\Subset \RR}$, and an endomorphism (respectively an automorphism) is
a representation on the same Hilbert space which maps $\A(I)$ into (respectively onto) itself.
Note that the family of bounded (connected) intervals is directed.
We shall denote by
$\gA_\A$ the associated quasi-local algebra, that is the
$C^*$-algebra
\begin{equation*}
\gA_\A:= \overline{\cup_{I\Subset\RR} \A(I)}
\end{equation*}
where the closure is meant in the operator norm topology. By the directedness,
any representation (resp. endomorphism, automorphism) of $\A|_\RR$ extends to a representation
(resp. endomorphism, automorphism) of the $C^*$-algebra $\gA_\A$.
Translations and dilations take bounded intervals $I \Subset \RR$ to
bounded intervals, hence
these transformations give rise to automorphisms of $\gA_\A$.

\subsection{KMS states on chiral nets: general remarks}\label{kmsstates}

In what follows we shall use the ``real-line'' picture.
A linear functional
$\psi:\gA_\A\to\CC$ such that its {\it local restriction}
$\psi|_{\A(I)}$ is normal for every bounded open interval $I\Subset
\RR$ is said to be {\bf locally normal} on $\gA_\A$.
Let now $\psi$ be a locally normal state on $\gA_\A$ and consider the
associated GNS representation $\pi_\psi$ of $\gA_\A$ on the Hilbert space
$\H_\psi$ with GNS vector $\Psi$. By construction, the vector
$\Psi$ is cyclic for the algebra $\pi_\psi(\gA_\A)$ and
$\langle \Psi, \pi_\psi(x)\Psi\rangle = \psi(x)$ for every
$x\in\gA_\A$.
\begin{lemma}
$\H_\pi$ is separable.
\end{lemma}
\begin{proof}
Let $I\Subset \RR$ be a bounded interval. The restriction
of $\pi_\psi|_{\A(I)}$ to the Hilbert space
$\overline{\pi_\psi(\A(I))\Psi}$ may be viewed as the GNS
representation of $\A(I)$ coming from the state
$\psi|_{\A(I)}$. It follows that $\overline{\pi_\psi(\A(I))\Psi}$
is separable, since (property $(9)$ in Section \ref{sec:prel1})
the local algebra $\A(I)$ is a type I\!I\!I$_1$ factor given on a separable
Hilbert space.

Let now $I_n:=(-n,n)\in\RR$ and $\H_{\psi,n}:=\overline{\pi_\psi(\A(I_n))\Psi}$
for every $n\in\NN$. Then, on one hand, $\H_{\psi,n}$
is separable for every $n\in\NN$; on the other hand, using that every
finite length interval $I$ is contained in {\it some} interval $I_n$,
it follows easily that $\cup_n \H_{\psi,n}$ is dense in $\H_\psi$.
Thus $\H_\psi$ is separable, as it is the closure of the union of a countable
number of separable Hilbert spaces.
\end{proof}
\begin{corollary}
The restriction of $\pi_\psi$ to any local algebra $\A(I)$ $(I\Subset \RR)$
is normal; thus $\A_\psi(I):=\pi_\psi(\A(I))$ is a von Neumann algebra on
$\H_\psi$, and $\pi_\psi|_{\A(I)}:\A(I)\to \A_\psi(I)$ is actually a unitarily
implementable isomorphism between type I\!I\!I$_1$ factors.
\end{corollary}
\begin{proof}
The listed facts follow from the last lemma since $\A(I)$ is
a type ${\rm {\!I\!I\!I}}_1$ factor given on a separable Hilbert
space.
\end{proof}

A translation of the real line takes every bounded interval into a
bounded interval. Thus the adjoint action of the strongly continuous
one-parameter group of unitaries $t\mapsto U(\t_t)$ associated to
translations, which is originally given for the chiral net $\A$,
may be viewed as a one-parameter group of
$*$-automorphisms of $\gA_\A$. Similarly, we may consider dilations,
too, as a one-parameter group $t\mapsto \Ad U(\d_t)$ of
$*$-automorphisms of $\gA_\A$.  We have that
\begin{equation*}
\Ad U(\tau_t)(\A(I)) = \A(t+I),;\;\;\;\;\; \Ad U(\d_t)(\A(I)) = \A(e^t I)
\end{equation*}
and we have the group relations
\begin{equation*}
\d_s\circ\tau_t = \tau_{e^s t} \circ\d_s.
\end{equation*}

Let $\a_t$ be a one-parameter automorphism group of the $C^*$-algebra $\gA_\A$.
A {\bf $\b$-KMS state} $\varphi$ on $\gA_\A$ with respect to $\a_t$ is a state with the following
condition: for any $x, y \in \gA_\A$ there is an analytic function $f$ on the strip $0 < \Im z <  \beta$, bounded and continuous on the closure of the strip, such that
\begin{equation*}
f(t) = \varphi(x\a_t(y)),\quad f(t+i\beta) = \varphi(\a_t(y)x).
\end{equation*}

In what follows we will be interested in states
on $\gA_\A$ satisfying the $\beta$-KMS condition
w.r.t. the one-parameter group
$t\mapsto \Ad U(\tau_t)$. As already said in the introduction, as a direct consequence
of the last recalled group-relations,
$\varphi$ is such a $\beta$-KMS state if and only if
$\varphi\circ \Ad U(\d_t)$ is a KMS state with {\it inverse
temperature} $\beta/e^t$. Thus it is enough to study
KMS states at the fixed inverse temperature $\beta=1$,
which we shall simply call a {\bf KMS state}.

A KMS state $\varphi$ of $\gA_\A$ w.r.t. $t\mapsto \Ad U(\tau_t)$
is in particular an invariant state for $t\mapsto \Ad U(\tau_t)$.
Thus, considering the GNS representation $\pi_\varphi$ associated to $\varphi$ on the
Hilbert space $\H_\varphi$ with GNS vector $\Phi$, we have that there
exists a unique one-parameter group of unitaries $t\mapsto V_\varphi(t)$
of $\H_\varphi$ such that
\begin{equation*}
V_\varphi(t)\pi_\varphi(x)\Phi = \pi_\varphi(\Ad U(\tau_t(x)))\Phi
\end{equation*}
for all $t\in \RR$ and $x\in\gA_\A$. It is well-known that $\Phi$
is automatically cyclic and separating for the von Neumann algebra
$\pi_\varphi(\gA_\A)''$ \cite{Takesaki-Winnink}, and that the associated modular group
$t\mapsto \Delta^{it}$ actually coincides with
$t\mapsto V_\varphi(t)$.

By the general result \cite[Theorem 1]{Takesaki-Winnink}, a KMS state is
automatically locally normal. Moreover, by \cite[Theorem 4.5]{Takesaki-Winnink}
every KMS state can be decomposed into {\it primary} KMS states. We recall
that a KMS state $\varphi$ is {\bf primary} iff it cannot be written as a
nontrivial convex combination of other KMS states and that it is
equivalent with the property that $\pi_\varphi(\gA_\A)''$ is a factor.

We also recall the KMS version of the well-known Reeh-Schlieder property.
Its proof relies on standard arguments,
see e.g.\! \cite[Prop.\! 3.1]{Longo}.
\begin{lemma}
Let $\varphi$ be a KMS state on $\gA_\A$ w.r.t. the one-parameter group
$t\mapsto\Ad U(\tau_t)$, and let $\pi_\varphi$ be the associated GNS representation with GNS vector
$\Phi$. Then $\Phi$ is cyclic and separating for $\pi_\varphi(\A(I))$ for every bounded
(nonempty, open) interval $I\Subset \RR$.
\end{lemma}

\subsection{The geometric KMS state}\label{geometric}

Here we show that every local, diffeomorphism covariant net $\A$ admits at least one KMS state, indeed this state has a geometric origin. The construction of this geometric KMS state $\geo$ is essential for our results, hence we include it in the present paper.

The geometric KMS state is constructed using two properties: Bisognano-Wichmann
property (valid also in higher dimensions), which implies that the
vacuum state is a KMS state for the $C^*$-algebra $\A\left(\RR_+\right)$
w.r.t. dilations; diffeomorphism covariance, by which it is (locally)
possible to find a map from $\RR$ to $\RR_+$ that sends translations
to dilations. Such a map would (globally) be the exponential, which
is not a diffeomorphism of $\RR$ onto $\RR$, but for any given interval we can find a diffeomorphism which coincides with the exponential map on that interval.

\begin{proposition}\label{pro:local-diffeom-geom-kms-state}
For any conformal net $\A$, there is a canonical injective endomorphism
$\Exp$ of the $C^*$-algebra $\mathfrak A_\A\equiv\overline{\bigcup_{I \Subset \RR} \A(I)}^{\|\cdot\|}$ such that
\begin{enumerate}
\item[(1)] $\Exp\left(\A\left(I\right)\right)=\A\left(e^{2\pi I}\right)$
\item[(2)] $\Exp\circ\Ad U(\t_{t})=\Ad U(\d_{2\pi t})\circ\Exp$,
\item[(3)] $\Exp$ is a $C^*$-algebra isomorphism of $\gA_\A$ with $\gA(\RR_+)\equiv
\overline{\bigcup_{I \Subset \RR_+} \A(I)}^{\|\cdot\|}$.
\end{enumerate}
\end{proposition}
\begin{proof}
For any $I\Subset\RR,$ choose a map $\eta_{I}\in C^{\infty}\left(\RR,\RR\right)$
such that: $\eta_{I}\left(t\right)=e^{2\pi t}$, $\forall t\in I$; outside
an interval $J\Subset\RR$ ($J$ has to contain both $I$ and $e^{2\pi I})$
$\eta_{I}$ is the identity map $\eta_{I}\left(t\right)=t$; $\eta_{I}^{-1}\in C^{\infty}\left(\RR,\RR\right)$.
Then $\eta_{I}$ is a diffeomorphism and has a unitary representative
$U\left(\eta_{I}\right)$ such that $\Ad U\left(\eta_{I}\right)\left(\A\left(J\right)\right)=\A\left(\eta_{I}J\right)$
and in particular $\Ad U\left(\eta_{I}\right)\left(\A\left(I\right)\right)=\A\left(e^{2\pi I}\right)$.
Set $\Exp|_{\A\left(I\right)}=\Ad U\left(\eta_{I}\right)$, this is a well-defined
endomorphism of $\cup_{I \Subset \RR} \A(I)$
(since $\Ad U\left(\eta_{I}\right)|_{\A\left(I\right)}=\Ad U\left(\eta_{J}\right)|_{\A\left(I\right)}$
whenever $I\subset J$) which can be extended to the norm closure $\gA_\A$
satisfying (1) and (3). Condition (2) follows from the corresponding relation
for maps of $\RR,$ $\eta_{I}\circ \t_{t}=\d_{2\pi t}\circ \eta_{I}$, and
the fact that, on every local algebra $\A\left(I\right)$,
\begin{multline*}
\Exp\circ\Ad U(\t_{t})=\Ad U\left(\eta_{I}\right)\circ\Ad U(\t_{t}) =\Ad U\left(\eta_{I}\circ \t_{t}\right) =\\
=\Ad U\left(\d_{2\pi t}\circ \eta_{I}\right) =\Ad U(\d_{2\pi t})\circ\Ad U\left(\eta_{I}\right) = \Ad U(\d_{2\pi t})\circ\Exp.
\end{multline*}
\end{proof}

\begin{theorem}\label{geometricconstruction}
For any conformal net $\A$, the state $\geo:=\w\circ\Exp$ is a primary
KMS state w.r.t. translations.\end{theorem}
\begin{proof}
By definition, the GNS representation of $\geo$ is (unitarily equivalent
to) the composition of the vacuum (identity) representation with $\Exp$:
$\left(\Exp,\H_{\Omega},\Omega\right)$. Thus $\pi_{\geo}\left(\gA_\A\right)^{\prime\prime}=\A\left(\RR_+\right)$
which is a factor: $\geo$ is a primary state.

The vector $\Omega$ is cyclic
and separating for $\A\left(\RR_+\right)$ and by the Bisognano-Wichmann property the
modular group is the group $t\mapsto U(\delta_{2\pi t})$
of (rescaled) dilations (dilations associated to the interval $\RR_+\subset S^{1},$
i.e. the ``true'' dilations), therefore
$\Ad\Delta_{\Omega}^{it}\circ\Exp=\Ad U(\d_{2\pi t})\circ\Exp=\Exp\circ\Ad U(\t_{t})$.

Hence, as the modular group w.r.t. $\Omega$ is the translation group
for the represented net $\RR\Supset I\mapsto\Exp\left(\A\left(I\right)\right)$,
the vector state $\Omega$ is a KMS state w.r.t. translations.
\end{proof}

\begin{rem}\label{essential-duality-geometric}
Consider the case where $\A$ is strongly additive. Then, in the vacuum representation of $\A$, we have
$\A\left(e^{2\pi a},\infty\right)\cap\A\left(e^{2\pi b},\infty\right)^{\prime}=\A\left(e^{2\pi a},e^{2\pi b}\right)$, therefore, by construction,
$\A_{\rm geo}\left(a,\infty\right)\cap\A_{\rm geo}\left(b,\infty\right)^{\prime}=\A_{\rm geo}\left(a,b\right)$,
for any $a<b<\infty$, where  $\A_{\rm geo}=\A_{\varphi_{{\rm geo}}}$ is defined as in eq.\!  \eqref{geo} here below.

%
\end{rem}

By the same arguments used in the proof of Theorem \ref{geometricconstruction}, we have the following.
\begin{proposition}\label{dil-tra}
There is a one-to-one map between the sets of
\begin{itemize}
\item KMS states on $\gA(\RR_+)\equiv \overline{\bigcup_{I \Subset \RR_+} \A(I)}^{\|\cdot\|}$ with respect to dilations
\item KMS states on $\gA_\A \equiv \overline{\bigcup_{I \Subset \RR} \A(I)}^{\|\cdot\|}$ with respect to translations.
\end{itemize}
The correspondence is given by $\varphi \mapsto \varphi\circ\Exp$.
\end{proposition}
By definition, the \emph{geometric} KMS state $\varphi_{\rm geo}$ of $\gA_\A$ is the KMS state corresponding to the vacuum state on $\gA(\RR_+)$ according to the above proposition: $\varphi_{\rm geo}\equiv \omega \circ \Exp$.

\section{The thermal completion and the role of relative
commutants}\label{thermalcompletion}

Let $\varphi$ be a locally normal state on the quasi-local algebra
$\gA_\A$ associated to a conformal net $(\A,U)$ and
$\pi_\varphi$ be the GNS representation with respect to $\varphi$.
For an $I\subset \RR$ we shall set
\begin{equation}\label{geo}
\A_\varphi(I)\equiv \{\mathop{\cup}_{I\supset
\tilde{I}\Subset\RR}{\pi_\varphi}(\A(\tilde{I}))\}''.
\end{equation}
Note that, when $I$ is a finite length (open) interval, $\A_\varphi(I)$ is
simply the image of $\A(I)$ under the representation ${\pi_\varphi}$; however,
$\A_\varphi$ is defined even for infinite length intervals.

Recall that representatives of local diffeomorphisms are contained in $\A$
(see Section \ref{diffeomorphismcovariance}).
Similarly as above, to simplify notations, for a diffeomorphism $\eta:\RR\to\RR$
localized in some finite length interval $I\Subset \RR$ we shall set
$U_\varphi(\eta):={\pi_\varphi}(U(\eta))$. The following basic properties can be easily
checked.

\begin{itemize}

\item
$\A_\varphi$ is {\it local} and {\it isotonous}: $[\A_\varphi(I_1),\A_\varphi(I_2)]=0$
whenever $I_1 \cap I_2 = \emptyset$ and $\A_\varphi(I_1)\subset \A_\varphi(I_2)$
whenever $I_1\subset I_2$.

\item
$U_\varphi(\eta)\A_\varphi(K)U_\varphi(\eta)^* = \A_\varphi(\eta(K))$
for every diffeomorphism $\eta$ localized in some finite length
interval and for every $K\subset \RR$.

\item
If $\A$ is {\it strongly additive}, then so is $\A_\varphi$: we have that
$\A_\varphi(r,t)\vee \A_\varphi(t,s) = \A_\varphi(r,s)$ for all $r<t<s$,
$r,t,s\in\RR\cup\{\pm \infty\}$.

\item
Assuming that $\A$ is strongly additive, if $\A_\varphi(\RR)={\pi_\varphi}(\gA_\A)''$
is a {\it factor}, then so are the algebras
$\A_\varphi(t+\RR_+), \A_\varphi(t+\RR_-)\; (t\in\RR)$, too (notice that $\A_\varphi(\RR_+) \cap \A_\varphi(\RR_+)^{\prime} \subset \A_\varphi(\RR_-)^{\prime} \cap \A_\varphi(\RR_+)^{\prime} = \left( \A_\varphi(\RR_-) \vee \A_\varphi(\RR_+) \right)^{\prime} = \A_\varphi(\RR)^{\prime}$).
\end{itemize}

Suppose $\varphi$ is a primary KMS state on $\gA_\A$ w.r.t.\! the
translations $t\mapsto \Ad U(\tau_t)$ and ${\pi_\varphi}$ is the GNS representation
associated to $\varphi$ with GNS vector $\Phi$. Then one can easily find
that $(\Phi,\A_\varphi(\RR_+)\subset \A_\varphi(\RR))$ is a standard half-sided
modular inclusion \cite{Wiesbrock,AZ} and, by the last listed property, it is actually an
inclusion of factors. In this situation, there
exists a unique (possibly not ``fully'' diffeomorphism covariant) M\"obius covariant, strongly additive net $(\hat{\A}_\varphi,\hat{U}_\varphi)$ such that
\begin{itemize}
\item $\hat{U}_\varphi(g)\Phi = \Phi$ for every M\"obius transformation $g$,
\item $\hat{\A}_\varphi (\RR_+) = \A_\varphi(\RR)$ and
$\hat{\A}_\varphi(1+\RR_+)=\A_\varphi(\RR_+)$.
\end{itemize}
The net $(\hat{\A}_\varphi,\hat{U}_\varphi)$ is called the {\bf thermal
completion} of $\A$ w.r.t.\! to the primary KMS state $\varphi$ and it was
previously studied in \cite{Longo, SW}\footnote{The notion of thermal completion
was proposed in [24] based on heuristic considerations.}.
One has that
\begin{equation}\label{eq:thermal-completion}
\hat{\A}_\varphi(e^{2\pi t},e^{2\pi s}) = \A^d_\varphi(t,s)
\end{equation}
where
\begin{equation*}
\A^d_\varphi(t,s) = \A_\varphi(t,\infty)\cap \A_\varphi(s,\infty)'
\;\;\;\;\; (t<s, \; t,s\in\RR\cup\{\pm \infty\}).
\end{equation*}
Note that $\A_\varphi(t,s)\subset \A^d_\varphi(t,s)$ and, by
Remark \ref{essential-duality-geometric}, if $\A$ is strongly
additive and $\varphi$ is the geometric KMS state, this inclusion is
actually an equality \footnote{We warn the reader that, in \cite{Longo},
the implication $(i)\Rightarrow (ii)$ in Prop.\! 3.5 and Cor.\! 3.6 are incorrect,
yet they have not been used in the sequel of that paper.}:
\begin{equation}\label{eq:dual-net-geometric-KMS}
\A^d_{\rm geo}(t,s) = \A_{\rm geo}(t,s).
\end{equation}

\begin{theorem}
Let $\A$ be a conformal net satisfying the split property,
$\varphi$ a primary KMS state on $\gA_\A$ with GNS representation
${\pi_\varphi}$, and assume that
\begin{equation}\label{a=ad}
\A_\varphi(t,s)= \A^d_\varphi(t,s)
\end{equation}
for some $t<s, \, t,s\in \RR$. Then
$\A$ is strongly additive and $\varphi$ is of the form
$\varphi=\geo \circ \alpha$ where $\geo$
is the geometric KMS state and
$\alpha\in {\rm Aut}(\gA_\A)$ such that
\begin{itemize}
\item
$\alpha(\A(I))=\A(I)$ for all $I\Subset\RR$
\item
$\alpha\circ \Ad U(\tau_t) = \Ad U(\tau_t)\circ \alpha$ for all $t\in\RR$.
\end{itemize}
In particular, in this case the thermal completion
and the original net in the vacuum representation, as M\"obius covariant nets,
are unitarily equivalent.
\end{theorem}
\begin{proof}
By (local) diffeomorphism covariance, if the assumption
regarding the
relative commutant holds for a particular $t<s,\, t,s\in\RR$, then it
holds for all such pairs. So fix $t_1<t_2<t_3,\, t_1,t_2,t_3\in\RR$;
then by the strong
additivity of the thermal completion we have that
\begin{equation*}
\A_\varphi(t_1,t_2)\vee\A_\varphi(t_2,t_3) =
\hat{\A}_\varphi(e^{2\pi t_1},e^{2\pi t_2})\vee\hat{\A}_\varphi(e^{2\pi t_2},e^{2\pi t_3}) =
\hat{\A}_\varphi(e^{2\pi t_1},e^{2\pi t_3}) = \A_\varphi(t_1,t_3).
\end{equation*}
Since $\pi_{\varphi,I} \equiv {\pi_\varphi}|_{\A(I)}$ is a unitarily implementable
isomorphism for any finite length interval $I\Subset \RR$,
the above equation shows that $\A$ is strongly additive.
A similar argument shows that the split property of $\A$
implies the split property of the thermal completion.

Consider the GNS representations $\pi_\varphi$ and $\pi_{\rm geo}$
and the thermal completions $\hat{\A}_\varphi$ and $\hat{\A}_{\geo}$
associated to $\varphi$ and $\geo$, respectively. By \eqref{eq:thermal-completion}, \eqref{eq:dual-net-geometric-KMS} and point (1) of Prop. \ref{pro:local-diffeom-geom-kms-state},
the thermal completion given by the geometric KMS state is equivalent to
the {\it (strongly additive) dual} of the original net in the vacuum representation,
so, in our case, simply to the original net (which is already strongly additive):
\begin{equation*}
\hat{\A}_{\rm geo}(e^{2\pi t},e^{2\pi s}) = \A^d_{\rm geo}(t,s)= \A_{\rm geo}(t,s)= \A(e^{2\pi t},e^{2\pi s}). 
\end{equation*}

Fix a nonempty, finite length open interval $I\Subset \RR$. Since
both $\pi_{\varphi,I}$ and $\pi_{{\rm geo},I}$ are unitarily implementable,
there exists a unitary $V$ such that
\begin{equation*}
{\rm Ad}(V)|_{\A_{\varphi}(I)}=\pi_{{\rm geo},I}\circ\pi_{\varphi,I}^{-1}
\end{equation*}
and one has that for all $t,s\in I$
\begin{multline*}
V\hat{\A}_\varphi(e^{2\pi t},e^{2\pi s})V^* = V\A^d_\varphi(t,s)V^* =V\A_\varphi(t,s)V^* = \\
= \pi_{{\geo},I}(\A(t,s))= \A_{\geo}(t,s) = \A^d_{\rm geo}(t,s) = \hat{\A}_{\geo}(e^{2\pi t},e^{2\pi s}).
\end{multline*}
Thus, by \cite[Thm.\! 5.1]{Weiner}, it follows that two
thermal completions are equivalent:
there exists a unitary operator $W$
such that $W\hat{\A}_\varphi(a,b)W^* = \hat{\A}_{\geo}(a,b)$
for {\it all} $a,b\in\RR$ and
$W\hat{U}_\varphi(g)W^*= \hat{U}_{\geo}(g)$ for all
M\"obius transformations $g$. (Note that this latter fact
implies that ${\rm Ad}(W)$ also connects the respective
vacuum states of the two thermal completions.)
Then, using that both
$\A_{\geo}^d(I) =\A_{\geo}(I)$
and $\A_\varphi^d(I) =\A_\varphi(I)$,
one sees that the automorphism
of $\A(I)$
\begin{equation*}
\alpha_I:= \pi_{{\rm geo},I}^{-1}\circ {\rm Ad}(W^*)\circ \pi_{\varphi,I}
\end{equation*}
is well-defined (i.e.\! $W^*\A_\varphi(I)W = \A_{\geo}(I)$)
for every $I\Subset \RR$. Moreover, it is also clear that
$\alpha_I = \alpha_K|_{\A(I)}$ whenever $I\subset K$, hence
that it defines an automorphism $\alpha$ of $\gA_\A$ which preserves
every local algebra $\A(I),\, I\Subset \RR$.

The fact that ${\rm Ad}(W)$ connects the relevant representations of
the M\"obius group shows that $\alpha$ commutes with the one-parameter
group of translations $t\mapsto \Ad U(\tau_t)$. Moreover, since ${\rm Ad}(W)$
also connects the vacuum states of the two thermal completions, one can
also easily verify that $\geo\circ\alpha = \varphi$.
\end{proof}

As will be shown by examples in the second part of this paper,
without the assumption of the previous
theorem the inclusion $\A_\varphi(t,s)\subset \A_\varphi^d(t,s) \equiv
\A_\varphi(t,\infty)\cap \A_\varphi(s,\infty)'$ is not necessarily an equality. We
shall now investigate the completely rational case.

\begin{lemma}\label{le:finiteindextc}
Let ${\pi_\varphi}$ be the GNS representation of a primary KMS state $\varphi$ on $\gA_\A$.
If $\A$ is completely rational, then
$\A_\varphi(t,s)\subset \A_\varphi^d(t,s)$
is a finite index irreducible
inclusion.
\end{lemma}
\begin{proof}
We noted at the beginning of this section some basic properties of $\A_\varphi$. In particular, the strong
additivity of $\A$ implies the strong additivity of $\A_\varphi$, hence $\A_\varphi(t,\infty) = \A_\varphi(t,s)\vee \A_\varphi(s,\infty)$ and the relative commutant of the inclusion in question is simply the
center of $\A_\varphi(t,\infty)$. On the other hand, when
our KMS state is primary, the algebra $\A_\varphi(t,\infty)$ is a factor.
So our inclusion is indeed irreducible:
\begin{multline*}
\A_\varphi^d(t,s) \cap \A_\varphi(t,s)'  \! = \!
(\A_\varphi(t,\infty)\cap \A_\varphi(s,\infty)')\cap \A_\varphi(t,s)' = \\
 = \A_\varphi(t,\infty)\cap (\A_\varphi(s,\infty)\vee \A_\varphi(t,s))' = \A_\varphi(t,\infty)\cap \A(t,\infty)' = \CC\1.
\end{multline*}
Let now $n,m\in\NN$ with $0<n<m$. Since locally ${\pi_\varphi}$ is a unitarily
implementable isomorphism, the index of the inclusion
\begin{equation}
\label{eq:basic_inclusion}
\N_{n,m}:=\A_\varphi(t,s)\vee \A_\varphi(s+n,s+m) \subset \A_\varphi(t,s+m)\cap
\A'_\varphi(s,s+n)=:\M_{n,m}
\end{equation}
is simply the so called $\mu$-index $\mu_\A$ of the completely rational
net $\A$. Now it is clear that, as $m$ increases, both sides of
(\ref{eq:basic_inclusion}) increase, whereas, as $n$ increases, both sides
of (\ref{eq:basic_inclusion}) decrease. So let us set
\begin{eqnarray*}
\noindent
\N_n:=\{\cup_{m>n}\N_{n,m}\}'', & & \M_n:=\{\cup_{m>n}\M_{n,m}\}'',
\;\;{\textrm{and in turn}}\;\; \\
\N:= \cap_n \N_n, & & \M:= \cap_n \M_n.
\end{eqnarray*}
Fixing the value of $n$ and considering the sequence of
inclusions $m\mapsto (\N_{n,m}\subset \M_{n,m})$, by \cite[Prop.\! 3]{KLM}
we have that there is an expectation $E_n:\M_n\to \N_n$ satisfying the
Pimsner-Popa inequality with constant $1/\mu$. Note that even
without {\it a priori} assuming the normality of $E_n$, this implies
that the index of $\N_n\subset \M_n$ is less or equal to $\mu$; see
\ref{pimsnerpopa}. Then in turn, considering the sequence $n\mapsto
(\N_n\subset \M_n)$, we find that the index of the inclusion $\N\subset
\M$ is also smaller or equal to $\mu$.

Now it is rather straightforward that $\N_n = \A_\varphi(t,s)\vee
\A_\varphi(s+n,\infty)$. Moreover, we have
$\cap_n \A_\varphi(s+n,\infty) =\CC
\1$
since the intersection in question is clearly in the center of the factor
$\A_\varphi(\RR)={\pi_\varphi}(\gA_\A)''$. It is not obvious whether the order
of the operations ``$\vee$'' and ``$\cap$'' can be inverted:
\begin{equation*}
\N = \cap_{n} \big(\A_\varphi(t,s)\vee \A_\varphi(s+n,\infty)\big)
\stackrel{?}{=} \A_\varphi(t,s)\vee \big(\cap_n
\A_\varphi(s+n,\infty)\big)= \A_\varphi(t,s).
\end{equation*}
We shall now show that using the split property the above equation can be
justified. Indeed, by the split property, there exists a pair of
Hilbert spaces $\H_1$ and $\H_2$ and a unitary operator
$W$ such that
\begin{equation*}
W\A_\varphi(t,s) W^* \subset \B(\H_1)\otimes \CC\1_{\H_2},
\;\;\;{\rm and}\;\;
W\A_\varphi(t-1,s+1)' W^* \subset \CC\1_{\H_1}\otimes \B(\H_2).
\end{equation*}
In particular, $W\A_\varphi(t,s) W^*  = \K\otimes \CC\mathbbm 1_{\H_2}$ for some
$\K\subset \B(\H_1)$. Now, if $n\geq 1$, then by locality the
algebras $\A_\varphi(s+n,\infty)$ and
$\A_\varphi(t-1,s+1)$ commute and hence
\begin{equation*}
W\A_\varphi(s+n,\infty) W^* \subset \CC\1_{\H_1}\otimes \B(\H_2)
\end{equation*}
implying that $W\A_\varphi(s+n,\infty) W^* = 1_{\H_1}\otimes \R_n$ for some
$\R_n\subset \B(\H_2)$. Since it holds that $\cap_n \A_\varphi(s+n,\infty)=\CC\1$,
we have that $\cap_n \R_n = \CC\1_{\H_2}$ and
\begin{equation*}
W(\cap_n \N_n) W^*=
\cap_n (W\N_nW^*) = \cap_n (\K\otimes \R_n)
= \K\otimes (\cap_n \R_n) = \K\otimes \1_{\H_2}
= W\A_\varphi(t,s)W^*
\end{equation*}
which justifies that $\N=\cap_n \N_n = \A_\varphi(t,s)$.
By a similar argument,
again relying on the split property,
we can also show that
$\M_n = \A_\varphi(t,\infty)\cap \A_\varphi(s,s+n)'$
and hence that
\begin{equation*}
\M = \cap_n \M_n = \A_\varphi(t,\infty)\cap \A_\varphi(s,\infty)' = \A_\varphi^d(t,s)
\end{equation*}
which concludes our proof.
\end{proof}

\begin{theorem}\label{thm:thermal-completion-extension}
Let ${\pi_\varphi}$ be the GNS representation of a primary KMS state $\varphi$ on $\gA_\A$.
If $\A$ is completely rational, then the thermal
completion $(\hat{\A}_\varphi,\hat{U}_\varphi)$,
as a M\"obius covariant net, is conformal and unitarily equivalent to an irreducible local
extension of the original net $(\A,U)$. Moreover, this extension is
trivial (i.e.\! coincides with the original net in the vacuum representation) if and only if
$\A_\varphi^d(t,s)= \A_\varphi(t,s)$
for some (and hence for all) $t<s, \, t,s\in \RR$.
\end{theorem}
\begin{proof}
First note that by strong additivity, for all $r\in (t,s)$, we have that
\begin{multline}
\label{A(r,s)'capA^d(t,s)}
\A^d_\varphi(t,s)\cap \A_\varphi(r,s)' = (\A_\varphi(t,\infty) \cap \A_\varphi(r,s)')\cap \A_\varphi(s,\infty)' = \\
=\A_\varphi(t,\infty) \cap (\A_\varphi(r,s)\vee\A_\varphi(s,\infty))' = \A_\varphi(t,\infty)\cap \A_\varphi(r,\infty)'
= \A^d_\varphi(t,r).
\end{multline}
Similarly, we have that $\A^d_\varphi(t,s)\cap \A(t,r)' =\A^d_\varphi(r,s)$, too.
Now consider the faithful normal state $\tilde{\varphi}_{\rm geo}\circ E$
on  $\A^d_\varphi(t,s)$, where $E:\A_\varphi^d(t,s)\to \A_\varphi(t,s)$
is the (unique) faithful normal expectation whose existence
is guaranteed  by Lemma \ref{le:finiteindextc} and the state $\tilde{\varphi}_{\rm geo}$
on $\A_\varphi(t,s)$ is defined by the formula
\begin{equation*}
\tilde{\varphi}_{\rm geo}(x):=\geo({\pi_\varphi^{-1}}(x)),
\end{equation*}
$\forall x \in \A_\varphi(t,s)$.
Note that the above formula indeed well-defines a faithful
normal state since ${\pi_\varphi}$ is locally an isomorphism.
Being a faithful normal state on $\A^d_\varphi(s_1,s_2)$, it gives rise
to a one-parameter group of modular automorphisms $t\mapsto
\tilde{\sigma}_t$. By construction, $t\mapsto \tilde{\sigma}_t$
preserves $\A_\varphi(s_1,s_2)$ and on this subalgebra it acts like its
modular group associated to the state $\tilde{\varphi}_{\rm geo}$.

Locally,
both ${\pi_\varphi}$ and $\pi_{\rm geo}$ (the GNS representations
associated to $\varphi$ and $\geo$, respectively) are isomorphisms and the algebras
$\pi_{\rm geo}(\A(s,r))$ are local algebras of the thermal completion
net $\hat{\A}_{\geo}$. Hence, by the Bisognano-Wichmann property, it follows
that
\begin{equation*}
\tilde{\sigma}_t(\A_\varphi(s_1,r)) = \A_\varphi(s_1,f_t(r))\;\;\; \textrm{and}\;\;\;
\tilde{\sigma}_t(\A_\varphi(r,s_2)) = \A_\varphi(f_t(r),s_2).
\end{equation*}
Note that the actual formula of the function $f_t:(s_1,s_2)\to (s_1,s_2)$ could be
easily worked out (we would then also need to take account of the fact
that, when passing to the thermal completion net, one needs to perform a
re-parametrization). However, in what follows, we shall not need a
concrete formula for $f_t$, so for simplicity of the discussion we leave
the expression in this abstract form. Note further that,
by eq.\! \eqref{A(r,s)'capA^d(t,s)}, our previous formula holds for
the dual algebras, too:
\begin{equation}\label{eq:modular-action-extension}
\tilde{\sigma}_t(\A^d_\varphi(s_1,r)) = \A^d_\varphi(s_1,f_t(r))\;\;\;
\textrm{and}\;\;\;
\tilde{\sigma}_t(\A^d_\varphi(r,s_2)) = \A^d_\varphi(f_t(r),s_2).
\end{equation}

Let now $\Phi$ be the GNS vector given by the state $\varphi$ in its GNS
representation ${\pi_\varphi}$. Since $\A^d_\varphi(s_1,s_2)$ is a local algebra of the
thermal completion net $\hat{\A}_\varphi$ and $\Phi$ is the vacuum-vector
of this net, the modular group of unitaries
$t\mapsto \Delta_{\Phi}^{it}$ associated to $(\Phi, \A^d_\varphi(s_1,s_2))$
also acts in a ``geometrical manner'' on
$\A^d_\varphi(s_1,r)$ and we have that
\begin{equation}\label{eq:modular-action-thermal}
\Delta^{it}_\Phi \A^d_\varphi(s_1,r) \Delta^{-it}_\Phi = \A^d_\varphi(s_1,f_t(r)).
\end{equation}

Consider the inclusion of factors $\A^d_\varphi(t,r_0)\subset \A^d_\varphi(t,s)$
for some fixed $t<r_0<s$. It becomes a standard
half-sided modular inclusion of factors both when it is considered with
the state $\tilde{\varphi}$ given by the vector $\Phi$ and with the state
$\tilde{\varphi}_{\rm geo}\circ E$.
Indeed, it has been already shown that is is a half-sided modular inclusion.
Standardness with respect to $\Phi$ follows from the
Reeh-Schlieder property for KMS states (see Section \ref{kmsstates}).
As for $\tilde{\varphi}_{\rm geo}\circ E$, let $\Phi'$
be the GNS vector in the GNS representation $\pi'$. The subspace generated
by $\pi'(\A_\varphi(t,s))$ and $\Phi'$ is equivalent to the representation
space with respect to $\geo$, hence it holds that
$\overline{\pi'(\A_\varphi(t,r_0))\Phi'} = \overline{\pi'(\A_\varphi(r_0,s))\Phi'}$
again by the Reeh-Schlieder property. Note that
$\overline{\pi'(\A^d_\varphi(t,r_0))\Phi'}
 = \overline{\pi'(\A^d_\varphi(t,r_0)\vee\A_\varphi(r_0,s))\Phi'}$,
since $\A^d_\varphi(t,r_0)$ commutes with $\A_\varphi(r_0,s)$ and
$\overline{\pi'(\A_\varphi(r_0,s))\Phi'}$ is already included in
$\overline{\pi'(\A^d_\varphi(t,r_0))\Phi}$. By strong additivity of $\A$,
$\A^d_\varphi(t,r_0)\vee\A_\varphi(r_0,s)$ includes
$\A_\varphi(t,s)$, in particular the representatives of local diffeomorphisms
supported in $(t,s)$. Therefore it holds that
$\A^d_\varphi(t,r_0)\vee\A_\varphi(r_0,s) = \A^d_\varphi(t,s)$
and this implies the
cyclicity of $\Phi'$ for $\A^d_\varphi(t,r_0)$. The cyclicity for
$\A^d_\varphi(r_0,t)$ can be proved analogously.

Thus we can construct two M\"obius
covariant nets.
Of course, the one constructed with $\tilde{\varphi}$ simply gives back the
thermal completion $\hat{\A}_\varphi$. The other one, constructed
with the help of $\tilde{\varphi}_{\rm geo}\circ E$, is easily seen to be
a local  extension of the net obtained by the inclusion
$(\tilde{\varphi}_{\rm geo},\A_\varphi(t,r_0)\subset \A_\varphi(t,s))$ which in
turn is
equivalent to the thermal completion obtained with $\geo$ and hence with
the original net $\A$ (in the vacuum representation).

However, as we have seen their modular actions in equations
\eqref{eq:modular-action-extension} \eqref{eq:modular-action-thermal},
both constructed nets will have
$\A^d_\varphi(t,r)$ as the local algebra corresponding to the interval
$(e^{2\pi t},e^{2\pi r})$ for all $r\in [t,s]$.
Furthermore, it turns out that the extension of $\A$ is split.
Indeed, we have already observed in Lemma \ref{le:finiteindextc} that
the inclusion is irreducible and of finite index, and the original net $\A$ is
completely rational by assumption. Then by \cite{Longo03} a finite
index extension is split as well.
Hence the two strongly additive split nets coincide
on all intervals $(e^{2\pi r_1},e^{2\pi r_2})$, with $t<r_1,r_2<s$, and thus
by an application of \cite[Thm.\! 5.1]{Weiner} they are equivalent.
At this point we can infer that the extension $\A^d$ is conformal.
Indeed, it includes $\A$ as a subnet, in particular its Virasoro subnet,
hence there is a representation of $\diffs1$. Local representatives of $\diffs1$
supported in $(t,s)$ act covariantly on $\A^d_\varphi(t,s)$. Any interval
in $S^1$ can be obtained from $(t,s)$ and an action of M\"obius group,
any local diffeomorphism acts covariantly. The group $\diffs1$ is generated
by local diffeomorphisms, hence diffeomorphism covariance holds.

We have obtained that the thermal completion constructed with $\varphi$ is a
local extension of the original net (in the vacuum representation).
If $\A^d_\varphi(t,s)=\A_\varphi(t,s)$, then
of course the extension is trivial. On the other hand, a completely
rational net cannot be equivalent to a nontrivial extension of itself
since we have the formula \cite[Prop.\! 24]{KLM} relating the
$\mu$-indices of the net and of the extension to the index of the
extension.
\end{proof}

A M\"obius covariant net for which the only irreducible local extension is
the trivial one (i.e.\! itself) is said to be a {\bf maximal} net.
Putting together the two presented theorems, the following conclusion can
be drawn.
\begin{corollary}\label{cor:maximal-net-KMS-automorphisms}
Let $\A$ be a conformal net and $\varphi$ a primary KMS state on its
quasi-local algebra $\gA_\A$ w.r.t.\! the translations $t\mapsto \Ad U(\tau_t)$.
If $\A$ is completely rational and maximal, then there exists an
automorphism $\alpha\in{\rm Aut}(\gA_\A)$ satisfying
\begin{itemize}
\item
$\alpha(\A)(I)=\A(I)$ for all $I\Subset\RR$
\item
$\alpha\circ\Ad U(\tau_t) = \Ad U(\tau_t)\circ \alpha$ for all $t\in\RR$
\end{itemize}
such that $\varphi=\geo\circ\alpha$ where $\geo$ is the geometric KMS
state.
\end{corollary}

\section{Uniqueness results}\label{uniquenessresults}

\subsection{Maximal completely rational nets}\label{maximal}
As seen in Section \ref{thermalcompletion}, any KMS state $\varphi$ on a completely
rational maximal net is a composition
of the geometric KMS state $\geo$ and an automorphism
$\a \in {\rm Aut}(\gA_\A)$ such that $\a\circ\Ad U(\tau_t)=\Ad U(\tau_t)\circ\a$
for all $t\in\RR$ and $\a(\A(I))=\A(I)$ for all $I\Subset \RR$.
 From now on, we simply call such $\a$ an {\bf automorphism of the net $\A|_\RR$}
commuting with translations. Here we study these automorphisms.

As noted in the introduction, among many examples, completely rational
nets are of particular interest. A completely rational net admits only
finitely many sectors \cite{KLM}.  In this subsection we will show
the uniqueness of KMS state in cases where the net is completely rational and maximal
with respect to extension. To obtain the uniqueness, we need
to connect automorphisms on $\RR$ and sectors (on $S^1$ by definition).
Proposition \ref{sectorauto} will demonstrate that there is a nice correspondence
between them.

Let us begin with simple observations on automorphisms which commute with
rotations or translations.
\begin{proposition}\label{uniquenessauto}
Let $\s_1$ and $\s_2$ be two automorphisms of the net $\A$ commuting with
rotations.
If they are in the same sector, namely if there is a unitary operator $W$
which intertwines $\s_1$ and $\s_2$, then actually $\Ad (W)$ is
an inner symmetry.

\end{proposition}
\begin{proof}
By the definition of inner symmetry, we have just to prove that the vacuum vector $\Omega$ is invarian for $W$.

For any local element $x$ of $\A$ it holds that
$W\s_1(x)W^* = \s_2(x)$. Since $\s_1$ is an automorphism and surjective,
this is equivalent to $\Ad(W)(x) = \s_2\s_1^{-1}(x)$.
By assumption $\s_1$ and $\s_2$ commute with rotations, so does $\Ad(W)$.

Let $L_0$ be the generator of rotations. The observation above implies
that $\Ad(W)\circ \Ad(e^{itL_0}) = \Ad(e^{itL_0})\circ \Ad(W)$, for $t \in \RR$, or, by setting
$L_0^\prime := W^*L_0W$, that $\Ad(e^{itL_0^\prime}) = \Ad(e^{itL_0})$.
Since the net is irreducible in the vacuum representation, this
in turn shows that $e^{itL_0^\prime}$ is a scalar multiple of $e^{itL_0}$.
Let us denote the scalar by $\lambda(t)$.

It is immediate that $t \mapsto \lambda(t)$ is a continuous homomorphism from the group $\RR$
to the group of complex numbers of modulus $1$. Thus it follows that $L_0^\prime = W^*L_0W = L_0 + \epsilon$
where $\lambda(t) = e^{it\epsilon}$.
On the other hand, by the positivity of energy, the spectrum of
$L_0$ is bounded below. But $L_0$ and $L_0^\prime$ must have the same spectrum
since they are unitarily equivalent, hence
$\epsilon$ must be $0$. Namely, $W$ commutes with $L_0$.
This implies in particular that $W$ preserves
$\Omega$, an eigenvector of $L_0$ with multiplicity one.
\end{proof}

\begin{proposition}\label{geometricgaugeinvariance}
If an automorphism $\a$ of $\A|_\RR$ preserves the vacuum state $\omega$,
then $\a$ commutes with any diffeomorphism and it preserves also
the geometric KMS state.
\end{proposition}
\begin{proof}
The second part of the statement follows immediately from the first part,
since the geometric state is a ``composition of the vacuum with diffeomorphism'',
as seen from the construction in Section \ref{geometric}.

To show the first part, we observe that $\a_I$ is implemented by a unitary
operator $W$, since it preserves the vector state $\omega$ and this implementation
does not depend on the interval $I$, by the Reeh-Schlieder property.
Since, for any $I\subset \RR$, $\A(I)$ is preserved by $\Ad(W)$, so is $\A(I')$ ($= \A(I)'$ by the Haag duality), where
$I'$ is the complementary interval on $S^1$. Any interval on $S^1$ is either of the form $I$ or $I'$ with $I\subset \RR$.
By \cite[Corollary 5.8]{CW}, $W$ commutes with all
the diffeomorphisms.
\end{proof}

By an analogous proof as Proposition \ref{uniquenessauto},
we easily obtain the following proposition for the net on the real line $\RR$.
\begin{proposition}\label{uniquenessautor}
Let $\a_1$ and $\a_2$ be two automorphisms of the net $\A|_\RR$ commuting with
translations.
If they are unitarily equivalent, then the unitary operator $W$ which
intertwines $\a_1$ and $\a_2$ implements an inner symmetry.
\end{proposition}

The following lemmas will serve to connect different KMS states and
inequivalent automorphisms.
\begin{lemma}\label{dilationvariance}
If a locally normal state $\psi$ on $\A|_\RR$ is invariant under dilation
$\Ad U(\d_s)$ with some $s \in \RR \setminus \{0\}$, then
$\psi$ is equal to the vacuum state $\omega$.
\end{lemma}
\begin{proof}
It is obvious that $\psi$ is invariant under $\d_{ns}, n \in \ZZ$. Hence
we may assume that $s > 0$.

Let us consider intervals $I_T = [-T,T]$. As noted in \cite[Lemma 4.1]{Weiner},
the norm-difference of restrictions $\psi|_{\A(I_T)}, \omega|_{\A(I_T)}$
tends to $0$ when $T$ decreases to $0$. On the other hand, $\psi$ and $\omega$
are invariant under $\Ad U(\d_{ns})$ by assumption and definition respectively.
Therefore the norm-difference on $\A(e^{ns}I)$ is the same as on $\A(I)$ by the
invariance. Namely,
\begin{multline*}
\left\| \psi|_{\A(I_T)} - \omega|_{\A(I_T)} \right\| =
\left\| \psi\circ\Ad U(\d_{ns})|_{\A(I_{e^{-ns}T})}  - \omega\circ\Ad U(\d_{ns})|_{\A(I_{e^{-ns}T})} \right\| = \\
 = \left\| \psi|_{\A(I_{e^{-ns}T})} - \omega|_{\A(I_{e^{-ns}T})} \right\| \to 0,
\end{multline*}
which shows that the two states are the same state when restricted to $I_T$. As
$T$ is arbitrary, they are the same.
\end{proof}

\begin{lemma}\label{dilationinequivalence}
Let $\a$ be an automorphism of $\A|_\RR$ commuting with translations.
Let us denote the ``dilated'' automorphism $\Ad U(\d_s) \circ \a \circ \Ad U(\d_{-s})$ by
$\a_s$.
If $\a$ does not preserve the vacuum state $\omega$, then
the automorphisms of the family $\{\a_s\}_{s\in\RR_+}$ are
mutually unitarily inequivalent.
\end{lemma}
\begin{proof}
By assumption $\omega \circ \a$ is different from
$\omega$. Thus Lemma \ref{dilationvariance} implies that
the states of the family $\{\omega\circ\a\circ\Ad U(\d_s^{-1})\}_{s\in\RR}$ are mutually
different. We recall that $\omega$ is invariant under dilations, hence
this family is the same as the family $\{\omega \circ \a_s\}_{s\in\RR}$.

It is immediate that all the automorphisms $\{\a_s\}_{s\in\RR}$
commute with translations.
Then, by Proposition \ref{uniquenessautor}, any two of such automorphisms are
unitarily equivalent if and only if they are conjugate by an inner symmetry.
If there were such a pair of automorphisms, then their compositions
with the vacuum state $\omega$ would be equal, but this
contradicts the observation in the first paragraph.
\end{proof}

Next we construct a correspondence from automorphisms on $\A|_\RR$
to automorphic sectors of $\A$.
\begin{proposition}\label{sectorauto}
For any automorphism $\a$ on $\A|_\RR$ which commutes with translations,
there corresponds an automorphism $\s_\a$ of $\A$ which commutes with rotations.
The images $\s_{\a_1}$ and $\s_{\a_2}$ are unitarily equivalent if and only if
$\a_1$ and $\a_2$ are unitarily equivalent.
\end{proposition}
\begin{proof}
Recall that the real line $\RR$ is identified with a subset of $S^1$
as explained in Section \ref{preliminaries}.
First we fix an open interval $I_0$ whose closure does not contain the
point at infinity and has the length $2\pi$ in the real line picture.
Note that $S^1 \setminus \{\infty\}$ is naturally diffeomorphic to an interval
$I_0$ of length $2\pi$. Indeed, there is a diffeomorphism
from $S^1\setminus\{\infty\}$ onto $I_0$ which preserves the lengths in the circle picture
of $S^1\setminus\{\infty\}$ with respect to the lengths in the real-line picture of $I_0$. Let us
call this diffeomorphism $\eta_0$. Let $p$ be a point in $S^1\setminus\{\infty\}$.
If $s_p > 0$ (or $s_p < 0$) is small enough so that for any $0 \le s^\prime \le s_p$
(or $0 \ge s^\prime \ge s_p$) it holds that $\rho_{s^\prime}(p) \in S^1\setminus\{\infty\}$,
then it is easy to see that $\eta_0\circ\rho_{s^\prime}(p) = \tau_{s^\prime}\circ\eta_0(p)$.

We have to define an automorphism $\s_\a$ through $\a$. Let us take an
interval $I \subset S^1$. We can choose a rotation
$\rho_s$ such that $\overline{\rho_s(I)}$ is inside $S^1\setminus\{\infty\}$.
It is again easy to see that there is a diffeomorphism $\eta$ of $S^1$ which
coincides with $\eta_0$ on $\rho_s(I)$.
The desired automorphism is defined by
\[
\s_{\a,I} := (\Ad(U(\rho_s)))^{-1} \circ (\Ad(U(\eta)))^{-1} \circ \a \circ
 \Ad(U(\eta)) \circ \Ad(U(\rho_s)).
\]
Since $\a$ preserves each algebra $\A(I)$ on any interval $I$, this is an automorphism.
We must check that this definition does not depend on $s$ and $\eta$ and that
$\s_{\a,I}$ satisfy the consistency condition w.r.t. inclusions of intervals.

Let us fix $s$ which satisfies the condition that $\rho_s(I)$ does not touch
the point at infinity. A different choice of $\eta$ under the condition that
$\eta$ coincides with $\eta_0$ on $\rho_s(I)$ does not matter at all.
Indeed, let $\eta^\prime$ be another diffeomorphism which complies with the
condition. Then $\eta^{-1}\circ\eta^\prime$ does not move points in
$\rho_s(I)$; in other words, the support of $\eta^{-1}\circ\eta^\prime$
is in the complement of $\rho_s(I)$. Since $U$ is a projective unitary representation,
it holds that $U(\eta^\prime) = c\cdot U(\eta)U(\eta^{-1}\circ\eta^\prime)$,
where $c$ is a scalar with modulus $1$, hence the adjoint actions of $U(\eta^\prime)$ and $U(\eta)$ on $\A(\rho_s(I))$ are the same
by the locality of the net.

We consider next different choices $s_1 < s_2$ of rotations. A rotation of $2\pi$ is
just the identity, thus we may assume that $s_2<2\pi$ and that, for any $s_1 \le s \le s_2$,
the interval $\rho_s(I)$ never contains $\infty$.
Then, for any point $p$ of $I$ and for any $0 \le t \le s_2-s_1$,
it holds that $\eta_0\circ\rho_{t}\circ\rho_{s_1}(p) = \tau_{t}\circ\eta_0\circ\rho_{s_1}(p)$.
The adjoint action of a diffeomorphism on $\A(I)$ is determined by the
action of the diffeomorphism on $I$ (by a similar argument to that in the previous
paragraph), so it holds that
\[
\Ad(U(\eta)) \circ \Ad(U(\rho_t)) |_{\A(\rho_{s_1}(I))}= \Ad(U(\tau_t)) \circ \Ad(U(\eta)) |_{\A(\rho_{s_1}(I))}.
\]
By assumption $\a$ commutes with $\Ad(U(\tau_t))$ for any $t$, hence, putting $t=s_2-s_1$, we have on $\A(I)$
\begin{multline*}
\Ad(U(\rho_{s_2}))^{-1} \circ \Ad(U(\eta))^{-1} \circ \a \circ  \Ad(U(\eta)) \circ \Ad(U(\rho_{s_2})) = \\
= \Ad(U(\rho_{s_1}))^{-1} \circ \Ad(U(\rho_{t}))^{-1} \circ \Ad(U(\eta))^{-1} \circ \a \circ  \Ad(U(\eta)) \circ \Ad(U(\rho_{t})) \circ \Ad(U(\rho_{s_1})) = \\
= \Ad(U(\rho_{s_1}))^{-1} \circ \Ad(U(\eta))^{-1} \circ \Ad(U(\tau_t))^{-1} \circ \a \circ \Ad(U(\tau_t)) \Ad(U(\eta)) \circ \Ad(U(\rho_{s_1})) = \\
= \Ad(U(\rho_{s_1}))^{-1} \circ \Ad(U(\eta))^{-1} \circ \a \circ \Ad(U(\eta)) \Ad(U(\rho_{s_1})).
\end{multline*}

This completes the proof of well-definedness of $\sigma_{\a,I}$.

Let us check the consistency w.r.t. inclusions of intervals. If $I \subset J$, then the $\eta$ and $\r_s$ chosen
for the larger interval $J$ still work also for $I$ and their action on $I$ is just a restriction.

To confirm that $\sigma_\a$ commutes with rotations, let us fix an interval $I$.
Let us choose $\eta$ and $s$ as above. If $t$ is small enough so that $\rho_t(\rho_s(I))$
does not touch $\infty$, then a similar calculation as above shows that
$\Ad(U(\rho_t))$ commutes with $\sigma_{\a,I}$. By repeating a small rotation we obtain arbitrary rotations.
We just have to check that the set of allowed $t$ above, for $\rho_{s'}(I)$ ($s'\in\RR$), depends on the
length of $I$ and not on the position of $\rho_{s'}(I)$. Indeed, for any $s'$, we can choose $s$ so that
$\rho_{s} \left(\rho_{s'}(I) \right)$ is at the same fixed distance from $\infty$.

Automorphisms on $S^1$ commuting with rotations (respectively on $\RR$ commuting with translations)
are unitarily equivalent if and only if they are conjugated by an inner symmetry
by Proposition \ref{uniquenessauto} (respectively Proposition \ref{uniquenessautor}).
An inner symmetry commutes with any diffeomorphism, on the other hand the correspondence
$\a \mapsto \sigma_\a$ is constructed with composition with diffeomorphisms. From this
it is immediate to see the last statement.
\end{proof}

Let us conclude this subsection with a uniqueness result for maximal rational nets.
\begin{theorem}\label{uniquenessmaximal}
If a net $\A$ is completely rational and maximal,
then it admits a unique KMS state, the geometric state $\geo$.
\end{theorem}
\begin{proof}
We have seen in Cor. \ref{cor:maximal-net-KMS-automorphisms} that any primary KMS
state $\varphi$ on such a net
is a composition of the geometric state with an automorphism $\a$ commuting with
translations.
Let us assume that $\varphi$ were different from the geometric state.
Then by Proposition \ref{geometricgaugeinvariance}, $\a$ must change
the vacuum state $\omega$. Then Lemma \ref{dilationinequivalence} would imply that
all the automorphisms $\{\a_s\}$ are mutually unitarily inequivalent.
 From these automorphisms we could construct mutually inequivalent sectors
by Proposition \ref{sectorauto}. This contradicts with the finiteness of
the number of sectors in a completely rational net. Thus if a KMS state
$\varphi$ is primary, then it is the geometric state.

An arbitrary KMS state is a convex combination of primary KMS states
\cite{Takesaki-Winnink}, hence in this case the geometric state itself.
\end{proof}

\subsection{General completely rational nets}\label{general}
Here we show the uniqueness of KMS state for general completely rational nets.
In the previous section we have proved that any maximal completely rational net
admits only the geometric state. One would naturally expect that, if one has an
inclusion of nets with finite index, then every KMS state on the smaller net
should extend to the larger net, thereafter the uniqueness would follow from
the uniqueness for maximal nets. Unfortunately the present authors are not aware
of such a general statement. Instead, we will see that if we have a KMS state
then its thermal completion admits some KMS state. We repeat this procedure
and arrive at the maximal net, where any KMS state is geometric, and find
that the initial state was in fact geometric as well.
\subsubsection{Extension trick}\label{extensiontrick}
Let $\A$ be a completely rational net and $\varphi$ be a KMS state on $\A$.
In this case, as we saw in Theorem \ref{thm:thermal-completion-extension},
the thermal completion ${\hat{\A}_\varphi}$ of $\A$ with respect to $\varphi$ is identified with
an extension of the net $\A$. The objective here is to construct another KMS state
on ${\hat{\A}_\varphi}$.

By Lemma \ref{le:finiteindextc} and eq.\! \ref{eq:thermal-completion},
$\A_\varphi(a,b) \subset {\hat{\A}_\varphi}(e^{2\pi a},e^{2\pi b})$ is an irreducible finite index inclusion for each interval $(a,b)$;
therefore, there is a unique conditional expectation
\begin{equation*}
E_{(a,b)}: {\hat{\A}_\varphi}(e^{2\pi a},e^{2\pi b}) \longmapsto \A_\varphi(a,b).
\end{equation*}
It is easy to see that this is a consistent
family w.r.t. inclusions of intervals. We denote simply by $E$ the map defined on the closed union
$\overline{\bigcup_{I\Subset \RR_+} {\hat{\A}_\varphi}(I)}^{\|\cdot\|}$.
Let us define the state
\begin{equation}\label{eq:KMS-cond-expectation}
\hat{\omega} = \omega \circ \Exp\circ {\pi_\varphi}^{-1} \circ E
\end{equation}
on $\overline{\bigcup_{I\Subset \RR_+} {\hat{\A}_\varphi}(I)}^{\|\cdot\|}$.
We will show that $\hat{\omega}$ is a KMS state with respect to
dilations. We collect general remarks in \ref{irreducibleinclusion} and \ref{kmscondition}.

First of all,
we recall that the original net $\A$ in the vacuum representation is diffeomorphism covariant.
Even in the GNS representation ${\pi_\varphi}$ with respect to $\varphi$, as explained at the beginning of Section \ref{thermalcompletion}, local diffeomorphisms act covariantly on each intervals, implemented by $U_\varphi$:
$U_\varphi(\eta)\A_\varphi(I)U_\varphi(\eta)^* = \A_\varphi(\eta(I))$. Since the extended net
$\A^d_\varphi$ is defined as the relative commutant
$\A^d_\varphi(a,b) := \A_\varphi(a,\infty)\cap\A_\varphi(b,\infty)^\prime$,
local diffeomorphisms $U_\varphi$ respect the structure of intervals:
\begin{equation*}
U_\varphi(\eta)\A^d_\varphi(a,b)U_\varphi(\eta)^* = \A^d_\varphi(\eta(a),\eta(b)).
\end{equation*}
In particular, if a diffeomorphism $\eta$ preserves an interval of finite length
$I$, then it acts on $\A^d_\varphi(I)$ as an automorphism.

On the original net $\A$, we know that the modular automorphism of $\A(I)$ with respect
to the vacuum $\omega$ acts as the dilation associated to $I = (a,b)$. On $\A(I)$ such dilation can be
implemented by local diffeomorphisms $\eta_t$. In fact, the dilation preserves $I$, hence
it is enough to modify this outside $I$ so that the support is compact.
If we restrict $\hat{\o} = \omega \circ \Exp \circ \pi_\varphi^{-1} \circ E$ to
$\A_\varphi(a,b)$, where $\Exp$ is defined in Prop.\! \ref{pro:local-diffeom-geom-kms-state}, the modular automorphism is
\[
(\pi_\varphi\circ \Exp^{-1}) \circ\Ad U(\d^{\exp I}_t)\circ (\Exp \circ \pi_\varphi^{-1}),
\]
where
$\d^{\exp I}_t$ is the dilation associated to
$(e^{2\pi a},e^{2\pi b})$.
Take diffeomorphisms $\eta_t$ with
the condition specified above and notice that, although $\exp$ and $\log$ are diffeomorphisms only locally, $\log \circ \eta_t \circ \exp$ are global diffeomorphisms.
It holds on $\A_\varphi(a,b)$ that
\begin{multline*}
(\pi_\varphi\circ \Exp^{-1}) \circ\Ad U(\d^{\exp I}_t)\circ (\Exp \circ \pi_\varphi^{-1})
= (\pi_\varphi\circ \Exp^{-1}) \circ\Ad U(\eta_t)\circ (\Exp \circ \pi_\varphi^{-1}) = \\
= \pi_\varphi\circ \Ad U(\log \circ \eta_t \circ \exp) \circ \pi_\varphi^{-1}
= \Ad U_\varphi(\log \circ \eta_t \circ \exp).
\end{multline*}
By Lemma \ref{modularextension}, we see that $\Ad(U_\varphi(\log \circ \eta_t \circ \exp))$
is the modular automorphism of $\hat{\A}_\varphi(I)$ with respect to
$\hat{\o}$.
Let us assume that there is a sequence of local diffeomorphisms $\zeta^{I_n}_t$ supported
in $\RR_+$ whose actions
on $I_n := [\frac{1}{n},n]$ are dilation by $e^t$.
The adjoint action $\Ad(U_\varphi(\log \circ \zeta^{I_n}_t \circ \exp))$ of diffeomorphisms
on a local algebra $\hat{\A}_\varphi(e^{2\pi a},e^{2\pi b})$
is determined by the action of $\zeta^{I_n}_t$ on $(e^{2\pi a},e^{2\pi b})$,
hence we can consider the limit of these adjoint actions and we denote
it by $\s_t$.

On the other hand, translation on $\A_\varphi(a,b)$ is implemented
by unitaries $V_\varphi(t)$ in this GNS representation (note that a translation is not a local
diffeomorphism, hence we cannot define the representative through ${\pi_\varphi}$).
This in turn shows that $\Ad(V_\varphi(t))$ takes ${\hat{\A}_\varphi}(e^{2\pi a},e^{2\pi b})$ to
${\hat{\A}_\varphi}(e^{2\pi (a+t)},e^{2\pi (b+t)})$, by recalling the definition of ${\hat{\A}_\varphi}$.

We show that the two actions $\Ad(V_\varphi(t))$ and $\s_t$
are the same even on the thermal completion ${\hat{\A}_\varphi}$. In fact,
these two actions take ${\hat{\A}_\varphi}(e^{2\pi a},e^{2\pi b})$ to ${\hat{\A}_\varphi}(e^{2\pi (a+t)},e^{2\pi (b+t)})$,
hence the composition $\Ad(V_\varphi(t)) \circ \s^{-1}_t$ is an automorphism of ${\hat{\A}_\varphi}(e^{2\pi a},e^{2\pi b})$
and $\s$-weakly continuous.
It is obvious that this composition acts identically on $\A_\varphi(a,b)$, by considering
the two actions in the original representation, and if $t=0$ it is the identity.
Then by Lemma \ref{finiteautos}, second statement, it is constant for all $t$.

\begin{proposition}\label{newkms}
The state $\hat{\omega}$ on $\overline{\bigcup_{I\Subset \RR_+} {\hat{\A}_\varphi}(I)}^{\|\cdot\|}$ defined in (\ref{eq:KMS-cond-expectation}) is a KMS state with respect to dilations.
\end{proposition}
\begin{proof}
To apply the general statement of Proposition \ref{kmsconvergence}
to the inclusion of factors
${\hat{\A}_\varphi}(\frac{1}{2},2) \subset {\hat{\A}_\varphi}(\frac{1}{3},3) \subset \cdots {\hat{\A}_\varphi}(\frac{1}{n},n) \subset \cdots$,
and $\hat{\omega}$, we need to confirm that for each interval $I \Subset \RR_+$
the action of the modular automorphisms of ${\hat{\A}_\varphi}(\frac{1}{n},n)$ with respect
to $\hat{\omega}$ (for sufficiently large $n$) on
${\hat{\A}_\varphi}(I)$ is *-strongly convergent and
the limit is normal.
As remarked above, the action of the modular automorphisms is implemented
by local diffeomorphisms $U_\varphi$ and by \ref{localdiffeomorphisms} we may assume
that these diffeomorphisms $\eta^{I_n}_t$ are smoothly convergent. Then the
representatives $U_\varphi(\log \circ \eta^{I_n}_t\circ \exp)$ are strongly convergent, hence
their adjoint actions are *-strongly convergent as well, and the limit is normal.

Moreover, in this way we find diffeomorphisms $\zeta^{I}_t = \lim_n \eta^{I_n}$
which appeared in the previous remarks. Thus, when $I$ tends to $(0,\infty)$,
the limit of these adjoint actions $\Ad(U_\varphi(\log \circ \zeta^{I}_t \circ \exp))$
is $\s_t$, which in turn is equal to $\Ad V_\varphi(t)$.
\end{proof}

\subsubsection{Proof of uniqueness}\label{proofofuniqueness}
We continue to use the same notations as in Section \ref{extensiontrick}.

\begin{lemma}\label{trivialitycriterion}
The extended state $\hat{\omega}$ is the vacuum if and only if
$\varphi$ is the geometric KMS state.
\end{lemma}
\begin{proof}
If $\varphi$ is geometric, then, as we saw in Section \ref{geometric}, we have
$\pi_{\geo} = \Exp$ and the conditional expectation $E$ is trivial. Hence
$\hat{\omega} = \omega \circ \Exp \circ \Exp^{-1} = \omega$.

Conversely, suppose that $\hat{\omega}$ is the vacuum of the extended net.
We note that
\[
\hat{\omega}|_{\A_\varphi(a,b)} = \omega \circ \Exp \circ {\pi_\varphi}^{-1}|_{\A_\varphi(a,b)} = \geo\circ{\pi_\varphi}^{-1}|_{\A_\varphi(a,b)},
\]
but the vacuum of the extended net is the vector state $\<\Phi, \cdot \Phi\>$;
when restricted to ${\A_\varphi(a,b)}$ we have $\<\Phi, \cdot \Phi\> = \varphi\circ{\pi_\varphi}^{-1}(\cdot)$.
Hence this in turn means that the initial KMS state $\varphi$ is in fact $\geo$.
\end{proof}

\begin{theorem}
Any completely rational net $\A$ admits only the geometric KMS state $\geo$.
\end{theorem}
\begin{proof}
Any completely rational net has only finitely many irreducible extensions with finite index.
Let us consider a sequence of conformal extensions $\A_1 := \A \subset \A_2 \subset \cdots \subset \A_n$,
where $\A_n$ is maximal. By the remarked finiteness of extensions, the number of such
sequences is finite. Let $N_\A$ be the length of the longest sequence. If $\A$ is maximal,
then $N_\A$ is $1$. We will show the theorem by induction with respect to $N_\A$.
For the case $N_\A = 1$ we have already proved the thesis in Theorem \ref{uniquenessmaximal}.

We assume that the proof is done for nets with $N_\A < k$.
Let $\varphi$ be a primary KMS state on $\A$, where $N_\A = k$.
We perform the thermal completion ${\hat{\A}_\varphi}$ with respect to $\varphi$.
If ${\hat{\A}_\varphi}$ is not a proper extension, the same reasoning as in Section \ref{maximal} shows
that $\varphi = \geo$. Hence we may assume that ${\hat{\A}_\varphi}$ is a proper extension of $\A$. Let
$\hat{\omega}$ be the KMS state on
$\overline{\bigcup_{I \Subset \RR_+} \hat{\A}_\varphi(I)}^{\|\cdot\|}$ with respect to dilations
of Prop. \ref{newkms}.
Recall that there is a one-to-one correspondence between KMS states on
the half-line with respect to dilation and KMS states on the real-line
with respect to translation (see Proposition \ref{dil-tra}).
By definition of $N$, $N_\A = k$ implies $N_{{\hat{\A}_\varphi}} < k$. It follows from the assumptions
of induction that ${\hat{\A}_\varphi}$ admits only one KMS state
on the half-line, hence $\hat{\omega}$ is the vacuum.
In this case Lemma \ref{trivialitycriterion} tells us that the primary KMS state $\varphi$
is the geometric state on $\A$.
An arbitrary KMS state is a convex combination of
primary states, hence it is necessarily geometric.
This concludes the induction.
\end{proof}

\subsection{The uniqueness of KMS state for extensions}
In this section we consider the following situation. Let $\A \subset \B$
be a finite-index inclusion of conformal nets. We assume that $\A$ admits
a unique KMS state. Any conformal net
has the geometric KMS state $\geo$, hence the unique state is this.
We will show that the geometric state on $\A$ extends only to the geometric
state on $\B$; in other words $\B$ admits a unique KMS state, too.

We note that the construction of the geometric KMS state works for any diffeomorphism covariant net (thus relatively local w.r.t. the Virasoro subnet).
The result in this section is true even if $\B$ is not necessarily local.
We will use this fact for the analysis of two-dimensional conformal nets in next section.
\begin{theorem}\label{uniext}
If $\A$ admits a unique KMS state and $\A \subset \B$ is of finite index,
then $\B$ admits a unique KMS state as well (which is again the geometric state).
\end{theorem}
\begin{proof}
Let $\varphi_0$ be the unique KMS state of $\A$, namely the geometric state of $\A$. By construction, with $E$ the unique conditional expectation of $\B$ onto $\A$, the geometric KMS state $\varphi$ of $\B$ satisfies
\[
\varphi = \varphi_0\circ E\ .
\]
Let $\psi$ be a KMS state on $\B$. By the uniqueness of the KMS state on $\A$ we have
\[
\psi|_\A =\varphi_0 = \varphi|_\A \ .
\]
Let $\l > 0$ be the Pimsner-Popa bound for $E$, we have
\[
\varphi(x) = \varphi_0\circ E(x) = \psi\circ E(x) \geq \lambda\psi(x)
\]
for all positive elements $x\in\B$. Therefore $\psi$ is dominated by $\varphi$. As $\varphi$ is extremal, being the geometric KMS state, we then have $\psi=\varphi$.
\end{proof}

\section{KMS states for two-dimensional nets}
\label{Sect2}
We begin to recall the basic definitions and properties of a conformal net on the two-dimensional Minkowski spacetime. We refer to \cite{KL2} for more details and proofs.

Let $\M$ be the two-dimensional Minkowski spacetime,
namely $\mathbb R^2$ equipped with the metric
$\text{d}t^2 - \text{d}x^2$. We shall also use the lightray coordinates
$\xi_{\pm}\equiv t \pm x$. We have the decomposition
$\M=\L_+\times\L_-$ where $\L_{\pm}=\{\xi:\xi_{\pm}=0\}$ are the two
lightrays. A \emph{double cone} $\O$ is a non-empty open subset of
of $\M$ of the form $\O=I_+\times I_-$ with $I_{\pm}\subset\L_{\pm}$
bounded open intervals; we denote by $\K$ the set of double cones.

The M\"{o}bius group $\psl2r$ acts on $\mathbb
R\cup\{\infty\}$ by linear fractional transformations, hence this
action restricts to a local action on $\mathbb R$.  We then
have a local (product) action of
$\psl2r\times \psl2r$ on
$\M=\L_+\times\L_-$.

A \emph{local M\"{o}bius covariant net} $\A$ on $\M$ is a map
\[
\A:\O\in\K\mapsto\A(\O)
\]
where the $\A(\O)$'s are von Neumann algebras on a fixed Hilbert
space $\H$, with the following properties:
\begin{itemize}

\item \emph{Isotony.} $\O_1\subset\O_2\implies
\A(\O_1)\subset\A(\O_2)$.

\item \emph{Locality.} If $\O_1$ and $\O_2$ are spacelike separated
then $\A(\O_1)$ and $\A(\O_2)$ commute elementwise (two points $\xi_1$
and $\xi_2$ are spacelike if $(\xi_1 -\xi_2)_+(\xi_1 -\xi_2)_- <0$).

\item \emph{M\"{o}bius covariance.} There exists a unitary representation
$U$ of $\overline{\psl2r}\times\overline{\psl2r}$
on $\H$ such that, for every double cone $\O\in\K$,
\[
U(g)\A(\O)U(g)^{-1} = \A(g\O),\quad g\in\U,
\]
with $\U\subset\overline{\psl2r}\times\overline{\psl2r}$
any connected neighborhood of the identity
such that $g\O\subset\M$ for all $g\in\U$. Here $\overline{\psl2r}$
denotes the universal cover of $\psl2r$.

\item \emph{Vacuum vector.} There exists a unit $U$-invariant vector
$\Omega$, cyclic for $\bigcup_{\O\in\K}\A(\O)$.

\item \emph{Positive energy.} The one-parameter unitary subgroup of $U$
corresponding to time translations has positive generator.

\end{itemize}

The net $\A$ promotes to a local net on the Einstein cylinder $\E=\RR\times S^1$, covariant w.r.t. a suitable cover of  $\overline{\psl2r}\times\overline{\psl2r}$.
We shall always assume our nets to be
\emph{irreducible}.

A \emph{local conformal net} $\A$ on $\M$ is a M\"{o}bius covariant
net such that the unitary representation $U$ extends to a
projective unitary representation of $\text{Conf}(\E)$, the group of global, orientation preserving conformal diffeomorphisms of $\E$. In particular
\[
U(g)\A(\O)U(g)^{-1} = \A(g\O),\quad g\in\U\ ,
\]
if $\U$ is a connected neighborhood of the identity of
$\text{Conf}(\E)$, $\O\in\K$, and $g\O\subset\M$ for all
$g\in\U$. We further assume that
\begin{equation}\label{loc}
U(g)XU(g)^{-1} = X,\quad g\in\Diff(\RR)\times\Diff(\RR)\ ,
\end{equation}
if $X\in\A(\O_1)$, $g\in\Diff(\RR)\times\Diff(\RR)$ and $g$ acts
identically on $\O_1$. We may check the conformal covariance on $\M$
by the local action of $\Diff(\RR)\times\Diff(\RR)$.

Given a M\"{o}bius covariant net $\A$ on $\M$ and a bounded interval
$I\subset\L_{+}$  we set
\begin{equation}\label{movers}
\A_{+}(I)\equiv \bigcap_{\O=I\times J}\A(\O)
\end{equation}
(intersection over all intervals $J\subset \L_-$), and analogously
define $\A_-$.
By identifying $\L_{\pm}$ with $\RR$ we then get two
local nets $\A_{\pm}$ on $\RR$, \emph{the chiral components of $\A$}.
They extend to local nets on $S^1$ which satisfy the axioms of M\"{o}bius covariant local nets,
but for the cyclicity of $\Omega$. We shall also denote $\A_{\pm}$
by $\A_R$ and $\A_L$.
By the Reeh-Schlieder theorem the cyclic subspace
$\H_{\pm}\equiv\overline{\A_{\pm}(I)\Omega}$
is independent of the interval $I\subset\L_{\pm}$ and $\A_{\pm}$
restricts to a (cyclic) M\"{o}bius covariant local net on the
Hilbert space $\H_{\pm}$. Since $\Omega$ is separating for every
$\A(\O)$, $\O\in\K$, the map $X\in\A_{\pm}(I)\mapsto
X\restriction\H_{\pm}$ is an isomorphism for any interval $I$, so we
will often identify $\A_{\pm}$ with its restriction to $\H_{\pm}$.
\begin{proposition}\label{tensor}
Let $\A$ be a local conformal net on $\M$.
Setting $\A_0(\O)\equiv \A_+(I_+)\vee\A_-(I_-)$, $\O=I_+\times I_-$,
then $\A_0$ is a  conformal, irreducible subnet of
$\A$. There exists a consistent family of vacuum preserving
conditional expectations $\epsilon_{\O}:\A(O)\to\A_0(\O)$ and the
natural isomorphism from the product $\A_+(I_+)\cdot\A_-(I_-)$ to the
algebraic tensor product $\A_+(I_+)\odot\A_-(I_-)$
extends to a normal isomorphism between $\A_+(I_+)\vee\A_-(I_-)$ and
$\A_+(I_+)\otimes\A_-(I_-)$.
\end{proposition}
Thus we may identify $\H_+\otimes\H_-$ with
$\H_0\equiv\overline{\A_0(\O)\Omega}$ and $\A_+(I_+)\otimes\A_-(I_-)$ with
$\A_0(\O)$.

Let $\A$ be a local conformal net on the two-dimensional Minkowski
spacetime $\M$. We shall say that $\A$ is \emph{completely rational}
if the two associated chiral nets $\A_\pm$ in \eqref{movers} are completely rational.

\begin{proposition}\label{cr2}
If $\A$ is completely rational the following three conditions hold:
\begin{itemize}
\item[$a)$] \emph{Haag duality on $\M$.} For any double cone $\O$ we
have $\A(\O)=\A(\O')'$. Here $\O'$ is the causal complement of $\O$
in $\M$
\item[$b)$] \emph{Split property.} If $\O_1 , \O_2 \in\K$ and the
closure $\bar\O_1$ of $\O_1$ is contained in $\O_2$, the natural map
$\A(\O_1)\cdot\A(\O_2)'\to \A(\O_1)\odot\A(\O_2)'$ extends to
a normal isomorphism $\A(\O_1)\vee\A(\O_2)'\to \A(\O_1)\otimes\A(\O_2)'$.
\item[$c)$] \emph{Finite $\mu$-index.} Let $E=\O_1\cup\O_2\subset\M$
be the union of two double cones $\O_1 , \O_2$ such that $\bar\O_1$
and $\bar\O_2$ are spacelike separated. Then the Jones index
$[\A(E')':\A(E)]$ is finite. This index is denoted by $\mu_\A$,
\emph{the $\mu$-index of $\A$}.
\end{itemize}
\end{proposition}
\proof
One immediately checks that the three properties $a), b), c)$ are satisfied for the two-dimensional net $\A_0 = \A_+\otimes\A_-$ which is completely rational. Then $\A$ is an irreducible extension of $\A_0$ (see \cite{KL2}) that must be of finite-index, and this implies that $\A$ satisfies $a), b), c)$ too, by the same arguments as in the chiral case, cf. \cite{KL}.
\endproof

With $\A$ a local conformal net as above, we consider the quasi-local $C^*$-algebra $\mathfrak A \equiv\overline{\cup_{\O\in\K}\A(\O)}$ (norm closure) and the time translation one-parameter automorphism group $\tau$ of $\mathfrak A$. We have
\begin{theorem}
If $\A$ is completely rational, there exists a unique KMS state $\varphi$ of $\mathfrak A$ w.r.t. $\tau$. $\varphi$ is the lift by the conditional expectation of the geometric KMS state of $\A_0$.
\end{theorem}
The proof of the theorem follows by the above discussion and Thm.\! \ref{uniext}.
One can easily see that $\varphi$ is a geometric state too. We need the following proposition.

\begin{proposition}
Let $\A_+$ $\A_-$ be translation covariant nets of von Neumann algebras on $\RR$ and $\A_0$ the associated net on the two-dimensional Minkowski spacetime: $\A_0(I_+\times I_-)\equiv \A_+(I_+)\otimes\A_-(I_-)$. If $\varphi_0$ is an extremal KMS state of $\A_0$ w.r.t. time translations, then $\varphi_0 = \varphi_+\otimes\varphi_-$, where $\varphi_\pm$ is an extremal KMS state of $\A_\pm$ w.r.t. translations.
\end{proposition}
\proof
Let $\pi_{\varphi_0}$ be the GNS representation of $\mathfrak A_{\A_0}$ w.r.t. $\varphi_0$ and consider the von Neumann algebras $\M_0\equiv \pi_{\varphi_0}(\mathfrak A_{\A_0})''$ and $\M_\pm\equiv \pi_{\varphi_0}(\mathfrak A_{\A_\pm})''$. As $\pi_{\varphi_0}$ is extremal KMS, $\M_0$ is a factor, so $\M_+$ and $\M_-$ are commuting subfactors.

Now the translation one-parameter automorphism group of $\mathfrak A_{\A_0}$ extends to the modular group of $\M_0$ w.r.t. (the extension of) $\varphi_0$ and leaves the subfactors $\M_\pm$ globally invariant. By Takesaki theorem, there exists a normal  $\varphi_0$-invariant conditional expectation $\varepsilon_\pm:\M_0\to\M_\pm$. With $x_\pm\in\M_\pm$ we have \[
\varphi_0(x_- x_+)= \varphi_0(\varepsilon_- (x_ -x_+)) =  \varphi_0(x_- \varepsilon_- (x_+))
= \varphi_0(x_-) \varphi_0(x_+)  = \varphi_-(x_-) \varphi_+(x_+)  \ ,
\]
because $\varepsilon_- (x_+)$ belongs to the center of $\M_-$, so $\varepsilon_- (x_+)= \varphi_0(x_+)$. This concludes the proof.
\endproof

As a consequence, if $\A_\pm$ are completely rational, then $\A_0$ admits a unique KMS state w.r.t. time translations and this state is given by the geometric construction.

\subsubsection*{Acknowledgment.}
We would like to thank Kenny De Commer for a useful discussion.

\appendix
\newcommand{\appsection}[1]{\let\oldthesection\thesection
  \renewcommand{\thesection}{Appendix \oldthesection}
  \section{#1}\let\thesection\oldthesection}

\appsection{Pimsner-Popa inequality and normality}\label{pimsnerpopa}

We discuss here some properties of finite-index expectation needed in the paper, cf. \cite{I} for related facts.

Suppose $\N \subset \M$ is an inclusion of von Neumann algebras and
$E:\M\to \N$ is an expectation. Let
\begin{equation*}
E= E_n + E_s
\end{equation*}
be the (unique) decomposition of $E$ into the sum of a normal and a singular
$\M\to \N$ positive map
(with $E_n$ standing for the normal part and $E_s$ for the singular part).
As is known, one of the
equivalent definitions of singularity is that for any $P$ nonzero
ortho-projection there is a nonzero subprojection $Q\leq P$ such that
$E_s(Q) = 0$.

\begin{lemma}
\label{bimod}
$E_n(AX)=AE_n(X)$ and $E_n(XA)=E_n(X)A$ for all $A\in \N$ and $X\in \M$.
\end{lemma}
\begin{proof}
Let $T,S\in\N$ with $TS= ST = \1$ and $\Phi(\cdot):= T \cdot T^*$.
Then $\Phi^{-1}(\cdot)= S \cdot S^*$ and both $\Phi$ and $\Phi^{-1}$
are faithful positive normal maps. It follows that
$\Phi\circ E_n \circ \Phi^{-1}$ is a normal positive map and it is
also clear that
$\Phi\circ E_s \circ\Phi^{-1}$ is a positive map. We shall now show that
this latter one is actually a singular map.

It is rather evident that if $E_s \circ\Phi^{-1}$ is singular then so is
$\Phi\circ E_s \circ\Phi^{-1}$.
So let $P\in\M$ be a nonzero
ortho-projection. Then
$\Phi^{-1}(P)=SPS^*$ is a nonzero positive operator so its spectral
projection $Q$ associated to the interval $[a/2,a]$ where $a=\|SPS^*\|$
is nonzero and we have that $SPS^*\geq (a/2) Q$. By singularity of $E_s$,
there exists a nonzero subprojection $Q_0\leq Q,\, Q_0\neq 0$ such that
$E_s(Q_0)=0$. Then $TQ_0T^*$ is a nonzero positive operator so again we
shall consider its spectral projection $R$ associated to the interval
$[b/2,b]$ where $b=\|TQ_0T^*\|$. Again, it is nonzero and we have that
$\Phi(Q_0)=TQ_0T^*\geq (b/2) R$. Putting together the inequalities,
we have
\begin{equation*}
R\leq \frac{2}{b}\Phi(Q_0)\leq \frac{2}{b}\Phi(Q)\leq
\frac{2}{b}
\frac{2}{a}\Phi(\Phi^{-1}(P)) = \frac{4}{ab} P
\end{equation*}
and it is easy to see that if for two ortho-projections $P_1,P_2$
the inequality $P_1\leq t P_2$ holds for {\it some} $t>0$, then
actually $P_1\leq P_2$. So we have that $R$ is a nonzero
subprojection of $P$, and since $E_s\circ\Phi^{-1}$ is a positive
map, by the listed inequality we also have that
\begin{equation*}
E_s\circ\Phi^{-1}(R) \leq
\frac{2}{b}
E_s\circ\Phi^{-1}(\Phi(Q_0)) = \frac{2}{b}E_s(Q_0) = 0.
\end{equation*}
Thus $E_s \circ\Phi^{-1}$ --- and hence
$\Phi\circ E_s \circ\Phi^{-1}$, too --- are indeed singular.
However,
\begin{equation*}
\Phi\circ E_n \circ \Phi^{-1} + \Phi\circ E_s \circ \Phi^{-1}
= \Phi\circ E \circ \Phi^{-1} = E
\end{equation*}
since $TE(SXS^*)T^* = TSE(X)S^*T^* = E(X)$ for all $X\in \M$ and $S\in \N$.
Hence, by the uniqueness of the
decomposition, we have that $\Phi\circ E_s \circ \Phi^{-1} =E_s$
and $\Phi\circ E_n \circ \Phi^{-1} = E_n$ or, equivalently,
$\Phi\circ E_n = E_n\circ \Phi$. So we have that
\begin{equation}
\label{pre-bimod}
TE_n(X)T^* = E_n(TXT^*)
\end{equation}
for all $X\in \M$.
Now let $A\in \N$ be a strictly positive element (i.e.\!
$0\notin{\rm Sp}(A)\subset \RR+$). Then $T:=A$
and $\tilde{T} := \1 + i A$ are invertible elements in $\N$
with bounded inverse and so equation (\ref{pre-bimod}) can be applied
for both. After a straightforward calculation we obtain that
for all $X\in \M$
\begin{equation*}
[A, E_n(X)] =E_n([A, X]),
\end{equation*}
where $[Y,Z]=YZ-ZY$ is the commutator.
On the other hand, replacing $\tilde{T}$ by
$\tilde{T}=\1 + A$ and repeating the previous
argument we also find that for all $X\in \M$
\begin{equation*}
\{A, E_n(X)\} = E_n(\{A,X\}),
\end{equation*}
where $\{Y,Z\}=YZ+ZY$ is the anti-commutator.
So actually we have shown that $E_n$ commutes with both taking
commutators and taking anti-commutators with an arbitrary strictly
positive operator $A\in \N$. Then the claimed bimodule property follows,
since the linear span of strictly positive elements is
dense in $\N$ and $E_n$ is normal.
\end{proof}

Let now $F:\M\to\N$ be a positive map satisfying a Pimsner-Popa type
inequality \cite{PP}; i.e.\! we suppose that there exists a $\lambda>0$ such that
\begin{equation*}
F(X^*X)  \geq \lambda X^* X
\end{equation*}
for all $X\in \M$. Now consider the decomposition $F=F_n+F_s$ into
the sum of a normal and a singular positive maps. $F_n$ must be faithful. Indeed, an easy argument
relying on the normality of $F_n$ shows that, if there is a positive nonzero element
which is annihilated by $F_n$, then there is also a nonzero ortho-projection $P$
which is annihilated by $F_n$. However, there is a subprojection $Q\leq P$,
$Q\neq 0$ such that on this subprojection also $F_s$ is zero. Thus
$F(Q) = F_n(Q)+ F_s(Q)=0$ in contradiction with the assumed inequality.
Actually we can say much more.

\begin{lemma}
\label{normalPP}
The normal part $F_n$ of $F$ satisfies the Pimsner-Popa inequality
with the same constant $\lambda$.
\end{lemma}
\begin{proof}
By assumption we know that $K:=F-\lambda \cdot{\rm id}$ is a positive map. Our goal is
to show that $\tilde{K}:=F_n - \lambda {\rm id} = K - F_s$ is also a positive map.
Since $\tilde{K}$ is evidently normal, it is enough to show that if $P\in
\M$ is an ortho-projection then $\tilde{K}(P)\geq 0$. So let $P\in \M$ be an
ortho-projection and
\begin{equation*}
\S:= \{Q\in\M | \, Q^2=Q=Q^*, Q\leq P, \tilde{K}(Q)\geq 0\}.
\end{equation*}
Now $\S$ can be viewed as a partially ordered set (with the ordering given by the
operator ordering) and, if $\{Q_\alpha\}$ is a chain in $\S$, then --- by the
normality of $\tilde{K}$ ---
$Q:=\vee_\alpha Q_\alpha$ is still an element of $\S$. Hence, by an application
of the Zorn lemma, there is a maximal element in $\S$; say $Q\in \S$ is such an
element.

If $Q=P$, we have finished. So assume by contradiction that $P-Q$ is
nonzero. Then there exists a nonzero subprojection $R\leq P-Q$ such that
$F_s(R)=0$. Hence $\tilde{K}(R) = K(R)-F_s(R) = K(R)$ and
\begin{equation*}
\tilde{K}(Q+R) = \tilde{K}(Q) + \tilde{K}(R) =
\tilde{K}(Q) + K(R) \geq \tilde{K}(Q) + \lambda R \geq 0,
\end{equation*}
implying that $Q+R\in \S$ in contradiction with the maximality of $Q$.
\end{proof}

Let us return now to discussing expectations $E:\M\to \N$ (not necessarily normal),  with normal-singular
decomposition $E=E_n+ E_s$.

\begin{theorem}\label{normal-part-is-expectation}
Suppose $E$ satisfies the Pimsner-Popa inequality with
constant $\lambda >0$. Then $Z:=E_n(\1)$ is a strictly positive
and hence invertible element in the center of $\N$ and $\tilde{E}:=
Z^{-1}E_n$ is a normal expectation from $\M$ to $\N$ satisfying the
Pimsner-Popa inequality with the same constant $\lambda >0$.
\end{theorem}
\begin{proof}
By Lemma \ref{bimod} we have that
\begin{equation*}
AZ = A E_n(\1) = E_n (A) = E_n(\1)A = ZA
\end{equation*}
for all $A\in \N$, showing that $Z$ is indeed a central element.

We may
estimate $Z$ from above by considering that $\1 = E(\1) = E_n(\1) +
E_s(\1) = Z + E_s(\1)$ and the fact that $E_s$ is a positive map. From below, we
may apply our previous lemma. Putting them together, we have
\begin{equation*}
\lambda^{-1}\1 \leq Z = \1 - E_s(\1) \leq \1.
\end{equation*}
One of the inequalities shows that $Z^{-1}$ is bounded, whereas the other
shows that $Z^{-1}\geq \1$ and so $Z^{-1}E_n$ still satisfies the Pimsner-Popa inequality with
the same $\lambda$. The rest
of the statement -- namely that $Z^{-1}E_n$ is a normal expectation -- follows easily
from the facts so far established in this appendix.
\end{proof}

Now it turns out that the normal part is in fact the expectation itself.
The following corollary has been announced in \cite{Popa} without proof.
The argument here is due to Kenny De Commer.
\begin{corollary}\label{automatic-normality-expectations}
If a conditional expectation $E:\M \to \N$ satisfies the Pimsner-Popa inequality
with the constant $\l > 0$, then any conditional expectation $F:\M\to\N$ is normal.
\end{corollary}
\begin{proof}
As we have seen in Theorem \ref{normal-part-is-expectation}, there is a normal
conditional expectation $\tilde{E}:\M\to\N$ which satisfies the Pimsner-Popa inequality
with the same constant $\l$. Let us suppose that there is another
conditional expectation $F$. To show that $F$ is normal, it is enough to see
that for a bounded increasing net $\{x_\a\}$ of positive elements in $\M$
it holds that $\lim_\a F(x_\a) = F(\lim_\a x_\a)$ in $\s$-weak topology.
In fact, by replacing $x_\a$ with $x-x_\a$, it is equivalent to
show that if $x_\a$ is decreasing to $0$, then $F(\lim x_\a) = F(0) = 0$.

By the Pimsner-Popa inequality for $\tilde{E}$, we have $x_\a\le \lambda^{-1}\tilde{E}(x_\a)$.
We apply $F$ to the both sides to obtain
\begin{equation*}
F(x_\a) \le F(\lambda^{-1}\tilde{E}(x_\a)) = \lambda^{-1}F(\tilde{E}(x_\a)) = \lambda^{-1}\tilde{E}(x_\a),
\end{equation*}
since the image of $\tilde{E}$ is contained in $\N$ and $F$ is an expectation $\M\to\N$.
The normality of $\tilde{E}$ implies that the right-hand side tends to $0$, so does the left-hand side.
This proves the normality of $F$.
\end{proof}

\appsection{Irreducible inclusion of factors}\label{irreducibleinclusion}
Here we collect some observations on irreducible subfactors with a conditional expectation.
Throughout this appendix, $\N \subset \M$ is an irreducible inclusion of factors,
$E$ is the unique conditional expectation from $\M$ onto $\N$, $\varphi$ is
a faithful normal state on $\N$ and $\hat{\varphi} = \varphi \circ E$.

\begin{lemma}\label{modularcommutation}
If $\a$ is an automorphism of $\M$ which preserves $\N$ and the restriction to $\N$
preserves $\varphi$, then $\a$ commutes with the modular
automorphism group $\s^{\hat{\varphi}}_t$.
\end{lemma}
\begin{proof}
Since $\a$ preserves $\N$, $\a\circ E \circ \a^{-1}$ is a conditional expectation from
$\M$ onto $\N$. By the irreducibility such a conditional expectation is unique, hence
$\a \circ E \circ \a^{-1} = E$, or $\a\circ E = E \circ \a$.
We claim that $\a$ preserves $\hat{\varphi}$. Indeed, we have
\begin{equation*}
\hat{\varphi}(\a(x)) = \varphi(E(\a(x))) = \varphi(\a(E(x))) = \varphi(E(x)) = \hat{\varphi}(x).
\end{equation*}
 From this it follows that $\a$ commutes with $\s^{\hat{\varphi}}_t$ (see, for example,
\cite[chapter VIII, Cor.\! 1.4]{Takesaki2}).
\end{proof}

We insert a purely group-theoretic observation.
\begin{lemma}\label{period}
Let $G$ be a group and $\pi:\RR \to G$ be a group-homomorphism.
If there exists $n \in \NN$ such that for any $t \in \RR$ it holds
that $\pi(t)^{m_t} = e$ for some $m_t \le n$
where $e$ is the unit element in $G$, then $\pi(t) = e$, in other words $\pi$ is trivial.
\end{lemma}
\begin{proof}
Let us assume the contrary, namely that there were a $t$ such that $\pi(t) \neq e$.
Then $\pi(\frac{t}{n!}) \neq e$, since otherwise $\pi(t) = \pi(\frac{t}{n!})^{n!} = e$.
But by assumption there exists $m_t \le n$ such that
\begin{equation*}
\pi\left(\frac{t}{n(n-1)\cdots \hat{m_t} \cdots 2\cdot 1}\right)
= \pi\left(\frac{t}{n!}\right)^{m_t} = e,
\end{equation*}
where $\hat{m_t}$ means the omission of $m_t$ in the product. This is a
contradiction because the $n(n-1)\cdots \hat{m_t} \cdots 2\cdot 1$-th power of
the left hand side is $\pi(t) \neq e$.
\end{proof}

\begin{lemma}\label{modularextension}
Let the inclusion $\N \subset \M$ have finite index.
If $\{\a_t\}$ is one-parameter group of automorphisms
of $\M$ which preserve $\N$ and if it holds that $\a_t|_\N = \s^\varphi_t$,
then $\a_t = \s^{\hat{\varphi}}_t$.
\end{lemma}
\begin{proof}
By Lemma \ref{modularcommutation}, $\a_s$ commutes with $\s^{\hat{\varphi}}_t$.
Hence $\b_t := \a_{-t}\circ\s^{\hat{\varphi}}_t$ is again a one-parameter group
of automorphisms of $\M$, preserving $\N$, and its restriction to $\N$ is trivial
by assumption.

We claim that the one-parameter automorphism $\{\b_t\}$ is inner. Once we know this,
the lemma follows since the implementing unitary operators should be in the relative commutant,
which is trivial for an irreducible inclusion.

Suppose the contrary, namely that there were a $t \in \RR$ such that $\b_t$ is outer.
Let $\pi$ be the natural homomorphism $\Aut(\M) \to \Out(\M)$.

We show that the order of $\pi(\b_t)$ is smaller than the index $[\M,\N]$.
Indeed, if $\pi(\b_t)$ has order $p > [\M,\N]$, then $\gamma: \ZZ_p \to \Aut(\M)$,
$\gamma(n) := \b_{nt}$ is an outer action of $\ZZ_p$ on $\M$. If $\pi(\b_t)$ has
infinite order, then $\gamma(n) := \b_{nt}$ is an outer action of $\ZZ$.
In any case, the subfactor $\B^\gamma \subset \B$ has the index larger than
$[\M,\N]$. But this is a contradiction, since we have $\N \subset \M^\gamma \subset \M$
and the index of $\M^\gamma \subset \M$ has to be smaller than or equal to $[\M,\N]$.

Having seen that the order of any element $\pi(\beta_t)$
is smaller than or equal to $[\M,\N]$,
we infer that $\pi(\beta_t)$ is the unit element in $\Out(\M)$ by Lemma \ref{period},
which means $\beta_t$ is inner for each $t$.
\end{proof}

Finally we put a simple remark on a group of automorphisms of
irreducible inclusion $\N \subset \M$ with finite index.
\begin{lemma}\label{finiteautos}
Let $G$ be the group of automorphisms of $\M$ which act identically on $\N$.
Then $|G| \le [\M,\N]$. In particular, if $\{\b_t\}$ is a continuous family of
such automorphisms, then it is constant.
\end{lemma}
\begin{proof}
Note that any nontrivial element in $G$ is outer. In fact, if it were inner,
it would be implemented by an unitary $U \in \M$ which commutes with $\N$, hence by
the assumed irreducibility of $\N \subset \M$ it must be scalar.
By considering the inclusion $\N \subset \M^G \subset \M$ we see
that the order of $G$ cannot exceed the index of $\N \subset \M$. The second
statement follows immediately.
\end{proof}

\appsection{KMS condition on locally normal systems}\label{kmscondition}
In the present work we consider KMS states on the quasilocal algebra of conformal nets
with respect to translations or dilations. The typical systems, treated
e.g. in \cite[Section 5.3.1]{BR2}, are $C^*$- or a $W^*$-dynamical systems, but they are not directly applicable to our case.
Indeed, the algebra concerned is the quasilocal $C^*$-algebra generated by local von Neumann
algebras; on the other hand, the automorphisms concerned are translations or dilations, which
are not norm-continuous. Although the modification is rather
straightforward, for the readers' convenience we give a variation of the standard results in \cite{BR2} in a form
applicable to conformal nets.

Let $\M_1 \subset \M_2 \subset \cdots \subset \M_n \subset \cdots$ be a growing sequence
of von Neumann algebras and $\gM$ be the ``quasilocal algebra'' $\overline{\bigcup_n \M_n}^{\|\cdot\|}$.
We consider a state $\varphi$ on $\gM$ which is normal and faithful on each $\M_n$, i.e.
``locally normal and locally faithful''. (When we state some property with the adverb ``locally'',
we mean that the property holds if restricted to each local algebra $\M_n$).
Let $\s^n$ be the modular automorphism of $\M_n$ with respect
to $\varphi$. We assume that, for each $k$, $\s^n_t(\M_k) \subset \M_{k+1}$ for
sufficiently small $t$ irrespective of $n > k$. We assume also that $\s^n$ converges
to some one-parameter automorphism $\sigma$ pointwise *-strongly, $\sigma_t$ is
a locally normal map for each $t$ and $t\mapsto \s_t$ is pointwise *-strongly continuous. Let us call such a dynamical system
a {\bf locally normal system}.
 From these definitions, it is easy to see that $\sigma$ preserves $\varphi$.

\begin{definition}
Suppose that $\gM$ is a $C^*$ (or a $W^*$) algebra, $\sigma$ is a norm (resp. $\s$-weakly)
continuous one-parameter group of automorphisms
and $\psi$ is a state (resp. a normal state) on $\gM$. If for any $x,y \in \M$
and any function $g$ on $\RR$
which is the Fourier transform of a compactly supported function it holds that
\begin{equation*}
\int g(t)\psi(x\s_t(y))dt = \int g(t+i\beta) \psi(\s_t(y)x)dt,
\end{equation*}
then we say that $\psi$ satisfies the {\bf smeared KMS condition} with respect to
$\sigma$.
\end{definition}
In each case, $C^*$-dynamical system or $W^*$-dynamical system,
the usual KMS condition is equivalent to the smeared condition \cite{BR2}.
We use the same term for a locally normal system as well.

\begin{lemma}\label{smeared}
The state $\varphi$ satisfies the smeared KMS condition with respect to $\s$.
\end{lemma}
\begin{proof}
For each $x,y \in \M_k$, $\varphi$ satisfies the smeared condition with respect to
$\s^n$ where $n \ge k$. Namely, it holds that
\begin{equation*}
\int g(t)\varphi(x\s^n_t(y))dt = \int g(t+i\beta) \varphi(\s^n_t(y)x)dt.
\end{equation*}
We assumed that, for a fixed $t$, $\s^n_t(y)$ converges strongly to $\s_t(y)$.
Then the condition for $\s$ follows by the Lebegues' dominated convergence theorem.

A general element in $\gM$ can be approximated from $\{\M_n\}$ by norm.
\end{proof}

We fix an element $y \in \M_n$ and define the analytic elements
\begin{equation}\label{eq:smooth-y}
y_\e := \int \s_t(y) \sqrt{\frac{\pi}{\e}}\exp\left(-\frac{t^2}{\e}\right) dt.
\end{equation}
s.t. $y_\e \rightarrow y$ *-strongly for $\e \rightarrow 0$.
These are well-defined as elements of $\gM$. Indeed, if we truncate the integral
to a compact interval, then the integrand lies in some local algebra and the integral
defines a local element. Such truncated integrals converge in norm because of the
Gaussian factor, hence define an element of the $C^*$-algebra.
\begin{lemma}\label{entire}
For any locally normal state $\psi$, $\psi(\s_t(y_\e))$ continues to
an entire function of $t$.
\end{lemma}
\begin{proof}
By the assumed local normality of $\psi$, for a truncated integral,
the integral and $\psi$ commute. The full integral is approximated by norm,
hence the full integral and $\psi$ commute as well. Namely, for $z \in \CC$, we have
\begin{equation*}
\psi\left(\int \s_t(y) \sqrt{\frac{\pi}{\e}}\exp\left(-\frac{(t-z)^2}{\e}\right)\right) dt
 = \int \psi(\s_t(y)) \sqrt{\frac{\pi}{\e}}\exp\left(-\frac{(t-z)^2}{\e}\right) dt.
\end{equation*}
The right hand side is analytic and the left hand side is equal to $\psi(\s_z(y_\e))$
when $z$ is real.
\end{proof}

\begin{lemma}\label{kmsforanalytic}
For $x,y\in \M_n$, there is an analytic function $f$ such that
\begin{equation*}
f(t) = \varphi(x\s_t(y_\e)), f(t+i\beta) = \varphi(\s_t(y_\e)x).
\end{equation*}
\end{lemma}
\begin{proof}
We define $f$ by the first equation. We saw that $f$ is entire in Lemma \ref{entire}.
By Lemma \ref{smeared}, for any $g$, $\hat{g} \in \mathscr{D}$, it holds that
\begin{multline*}
\int g(t+i\beta)\varphi(\s_t(y_\e)x)dt = \int g(t)\varphi(x\s_t(y_\e)) dt = \\
= \int g(t)f(t)dt = \int g(t+i\beta)f(t+i\beta)dt.
\end{multline*}
Since $g$ is arbitrary under the condition above, we obtain the second equation.
\end{proof}

\begin{lemma}\label{uniform}
For $x,y \in \M_n$, $\varphi(x\s_t(y_\e))$ (respectively $\varphi(\s_t(y_\e)x)$) converges to
$\varphi(x\s_t(y))$ (respectively $\varphi(\s_t(y)x)$) uniformly on $t$.
\end{lemma}
\begin{proof}
We just prove the first, since the second is analogous by the assumed
*-strong convergence of the modular automorphisms.
Note that, by the Schwarz inequality and by the invariance of $\varphi$ with respect to $\s$,
we have
\begin{equation*}
\|\varphi(x\s_t(y_\e-y))\|^2 \le \varphi(x^*x)\varphi\left((y_\e-y)^*(y_\e-y)\right),
\end{equation*}
hence the uniformity is not a problem once we show the convergence of the right hand side.

By hypothesis, there is a $\d>0$ s.t. $\s_t(M_n)\subset M_{n+1}$ for $|t|\leq\d$.
Let us define $\tilde{y}_\e$ by the truncation of the integral in \eqref{eq:smooth-y}
to the subset $[-\d,\d]\subset\RR$. It follows that $\tilde{y}_\e\in M_{n+1}$,
$\left\|\tilde{y}_\e\right\| \leq \left\|y\right\|$ and, as the norm difference $\left\| \tilde{y}_\e - y_\e\right\|$ tends to $0$, it is enough to show the convergence of the right hand side with the local elements $\tilde{y}_\e$ in place of $y_\e$.
The restriction of  $\varphi$ to $M_{n+1}$ is normal and can be approximated in norm by linear
combinations of weakly continuous functionals of the form $\<\xi,\cdot \, \eta\>$ with a pair of vectors $\xi,\eta$.
Since $\<(y-\tilde{y}_\e)\xi,(y-\tilde{y}_\e)\eta\>$ is convergent to $0$ and the sequence $\tilde{y}_\e$ is bounded
the desired convergence follows.
\end{proof}

\begin{proposition}\label{kmsconvergence}
The state $\varphi$ satisfies the KMS condition with respect to $\s$.
\end{proposition}
\begin{proof}
As we saw in Lemma \ref{kmsforanalytic}, the KMS condition is satisfied
for any pair $x, y_\e$ where $x,y \in \M_n$. As $\e$ tends to $0$,
the analytic function $\varphi(x\s_t(y_\e))$ tends to $\varphi(x\s_t(y))$
uniformly on the strip by Lemma \ref{uniform} and by the three-line theorem.
The limit function connects $\varphi(x\s_t(y_\e))$ and $\varphi(\s_t(y_\e)x)$.
Any pair of elements in $\gM$ can be approximated in norm by elements in $\M_n$, hence the
same reasoning completes the proof.
\end{proof}

\appsection{Remarks on local diffeomorphisms}\label{localdiffeomorphisms}
We consider diffeomorphisms of $\RR$. We say simply a sequence of diffeomorphisms $\{\eta_n\}$
converges smoothly to a diffeomorphism $\eta$ when $\{\eta_n\}$ and all their
derivatives converge to $\eta$ uniformly on each compact set. Recall that
any diffeomorphism is a smooth ($C^\infty$) function $\RR \to \RR$ with strictly positive derivative.

\begin{lemma}\label{compacttranslation}
For each interval $I$, there is a diffeomorphism $\tilde{\t}_s$ with compact support
which coincides with translation $\tau_s$ on $I$.
\end{lemma}
\begin{proof}
We may assume $s = 1$.
There is a smooth non-negative function with a compact support whose value is strictly less than $1$.
By dilating this function, we may assume that its integral over $\RR$ is $1$.
By considering its indefinite integral, we obtain a smooth non-negative function
which is $0$ on $\RR_-$ and $1$ on some half-line $\RR_+ +a$, $a > 0$, with derivative
strictly less than $1$.
Similarly we obtain a smooth non-negative function which is $1$ on $\RR_-$ and
$0$ on $\RR_+ + a$ with derivative strictly larger than $-1$.
By translating and multiplying these functions, we obtain a non-negative function with
compact support with derivative larger than $-1$ which is $1$ on $I$.
The desired diffeomorphism is the function represented by this function added by
the identity function $\id(t) = t$.
\end{proof}

\begin{lemma}\label{compactsequencetr}
If a sequence of diffeomorphisms $\eta_n$ of $\RR$ converges smoothly
to translation $\tau_s$, then for any interval $I$ there is an
interval $\tilde{I} \supset I$ and a smoothly convergent sequence of diffeomorphisms $\tilde{\eta}_n$
with support in $\tilde{I}$ which coincides with $\eta_n$ on $I$ (hence converges
smoothly on $I$ to $\tau_s$).
\end{lemma}
\begin{proof}
Note that $\eta_n \circ \tau_{-s}$ converges smoothly to the identity map $\id$.
Let $g_n$ be functions which represent $\eta_n \circ \tau_{-s}$.
And $h$ be a function with a compact support such that $h(t) = 1$ on $I$.
Let us define
\begin{equation*}
\hat{g}_n(t) = (g_n(t) - t)h(t) + t.
\end{equation*}
Since $\{g_n\}$ converges to $\id$ smoothly, for sufficiently large $n$
their derivatives are strictly positive and define diffeomorphisms $\hat{\eta}_n$.
The function $\hat{g}_n$ coincides with $g_n$ on $I$ by the
definition of $h$. Let $\tilde{\tau}_s$ be the local diffeomorphism
constructed in Lemma \ref{compacttranslation}.
The composition $\tilde{\eta}_n := \hat{\eta}_n \circ \tilde{\tau}_s$
gives the required sequence.
\end{proof}

By the exponential map (or by an analogous proof) we obtain the
corresponding construction for dilation.
\begin{lemma}\label{compactsequencetdl}
If a sequence of diffeomorphisms $\eta_n$ of $\RR_+$ converges smoothly
to dilation $\d_s$, then for any interval $I \Subset \RR_+$ there is an
interval $\tilde{I} \supset I$ and a smoothly convergent sequence of diffeomorphisms $\tilde{\eta}_n$
with support in $\tilde{I}$ which coincides with $\eta_n$ on $I$ (hence converges
smoothly on $I$ to $\d_s$).
\end{lemma}

We apply these to the case of dilations of intervals.
The standard dilation (restricted to $\RR_+$) is the map $\d_s: \RR_+ \ni t \mapsto e^s t \in \RR$.
A dilation $\d^I_s$ of an interval $I$ is defined by $(\eta^I)^{-1} \circ \d^I_s \circ \eta^I$,
where $\chi^I$ is a linear fractional transformation which maps $I$ to $\RR_+$.
This is well-defined, since any other such linear fractional transformation is
a composition of the $\chi^I$ and a standard dilation.
\begin{lemma}
If $I_1 \subset I_2 \subset \cdots \subset I_n \subset \cdots \subset \RR_+$ is an increasing
sequence of intervals with $\bigcup_n I_n = \RR_+$, then for any fixed $s$,
$\{\d^{I_n}_s\}$ smoothly converge to $\d_s$.
\end{lemma}
\begin{proof}
Let us put $I_n = (a_n,b_n)$, hence $a_n \to 0$ and $b_n \to \infty$.
We take the fractional linear transformations as follows:
\begin{equation*}
\chi^{I_n}(t) = \frac{t-a_n}{b_n-t}, (\chi^{I_n})^{-1}(t) = \frac{b_n t + a_n}{t+1}.
\end{equation*}
Then we can calculate the dilation of $I_n$ concretely:
\begin{multline*}
\d^{I_n}_s(t) = (\chi^{I_n})^{-1} \circ \d^I_s \circ \chi^{I_n} = \frac{e^s b_n(t-a_n) + a_n(b_n-t)}{e^s(t-a_n)+b_n-t} = \frac{e^s(t-a_n) + a_n(1-\frac{t}{b_n})}{1+\frac{e^s(t-a_n)-t}{b_n}}.
\end{multline*}
 From this expression it is easy to see that $\d^{I_n}_s(t)$ converge smoothly
to $\d_s(t) = e^s t$, since the numerator tends smoothly to $e^s t$
and the denominator tends to $1$ smoothly.
\end{proof}

We summarize these remarks to obtain the following.
\begin{proposition}\label{intervaldilations}
For each $s$ and $I \Subset \RR_+$, there is a $\tilde{I} \Subset \RR_+$ and
a smoothly convergent sequence of diffeomorphisms $\eta^{I_n}_s$ with support in $\tilde{I}$
which converge to $\d_s$ and coincide with $\d^{I_n}_s$ on $I$.
\end{proposition}

\end{document}